%% file: concur.tex
\documentclass[a4paper,UKenglish,cleveref, autoref, thm-restate, dvipsnames]{lipics-v2021}
\hideLIPIcs  %uncomment to remove references to LIPIcs series (logo, DOI, ...), e.g. when preparing a pre-final version to be uploaded to arXiv or another public repository

\title{Algebraic Characterization of FO-definable Languages of Higher-Dimensional Automata} %TODO Please add

\titlerunning{Algebraic Characterization of FO-definable Languages of HDA}

\author{Enzo Erlich}{Université Paris Cité, CNRS, IRIF, F-75013, Paris, France \and EPITA Research Laboratory (LRE), Paris, France}{erlich@irif.fr}{https://orcid.org/0009-0007-3574-6362}{}

\author{Jérémy Ledent}{Université Paris Cité, CNRS, IRIF, F-75013, Paris, France}{jeremy.ledent@irif.fr}{https://orcid.org/0000-0001-7375-4725}{}

\author{Krzysztof Ziemiański}{Institute of Mathematics, University of Warsaw, Warsaw, Poland}{ziemians@mimuw.edu.pl}{https://orcid.org/0000-0001-7695-4028}{}

\authorrunning{E. Erlich, J. Ledent, K. Ziemiański} %TODO mandatory. First: Use abbreviated first/middle names. Second (only in severe cases): Use first author plus 'et al.'

\Copyright{Enzo Erlich, Jérémy Ledent and Krzysztof Ziemiański} %TODO mandatory, please use full first names. LIPIcs license is "CC-BY";  http://creativecommons.org/licenses/by/3.0/


 %TODO mandatory; please add comma-separated list of keywords

\category{} %optional, e.g. invited paper

\acknowledgements{}%optional

\nolinenumbers %uncomment to disable line numbering

\usepackage{mathpartir}
\usepackage{macros}

\EventEditors{Ana Sokolova and Patrick Totzke}
\EventNoEds{2}
\EventLongTitle{37th International Conference on Concurrency Theory (CONCUR 2026)}
\EventShortTitle{CONCUR 2026}
\EventAcronym{CONCUR}
\EventYear{2026}
\EventDate{September 1--4, 2026}
\EventLocation{Liverpool, UK}
\EventLogo{}
\SeriesVolume{391}
\ArticleNo{1}
\begin{document}

\maketitle

\begin{abstract}
    Higher-dimensional automata (HDA) are a model of concurrency that models simultaneous execution of events using higher dimensional cells.
    HDA recognize languages of pomsets, a generalization of finite words whose letters are partially ordered.
    We prove a new algebraic characterization of HDA languages: a language of pomsets is regular if and only if it is the inverse image of a functor from the category of pomsets into a finite category.
    Furthermore, the language is definable in first-order logic exactly when it is recognized by an aperiodic category, generalizing the McNaughton-Papert theorem to HDA languages.
    We also investigate a notion of counter-free HDA, and show that if a language is accepted by a counter-free HDA, it must be definable in first-order logic. The converse, however, is still open.  
\end{abstract}

\input{introduction}
\input{prelim}

\input{recognizability}
\input{firstorder}
\input{aperiodic}

\input{conclusion}

%%
%% Bibliography
%%

%% Please use bibtex, 

\bibliography{bibliography}

%\appendix
%
%\input{appendix}

\end{document}

%% file: introduction.tex
\section{Introduction}

Over finite words, the class of languages definable in First-Order~(FO) logic is a strict subclass of regular languages, which admits several characterizations.
For a language $L \subseteq \Sigma^*$, the following are equivalent:
\begin{enumerate}
\item $L$ is recognized by a star-free regular expression,
\item $L$ is recognized by an aperiodic monoid,
\item $L$ is recognized by a counter-free automaton,
\item $L$ is definable in first-order logic (FO),
\item $L$ is definable in linear temporal logic (LTL).
\end{enumerate}
The proof of $(1) \Leftrightarrow (2)$ is due to Sch\"utzenberger~\cite{Schutzenberger65}; the equivalence between $(2), (3)$ and $(4)$ was proved by McNaughton and Papert~\cite{McNaughton71}; and Kamp~\cite{Kamp68} showed the last equivalence $(4) \Leftrightarrow (5)$.
%This is due to the seminal work of Sch\"utzenberger~\cite{Schutzenberger65}, Kamp~\cite{Kamp68}, and McNaughton and Papert~\cite{McNaughton71}.

In this paper, we are interested in extending these results to \emph{Higher-Dimensional Automata} (HDA).
HDA are a model of concurrency that enriches standard finite-state automata with higher-dimensional cells, allowing to specify that some events may occur simultaneously.
Although they were originally introduced by Pratt~\cite{Pratt91} in 1991, a proper language theory for HDA was only developed recently~\cite{FahrenbergJSZ21}.
Many classic results of automata theory carry over to HDA languages~\cite{FahrenbergJSZ24, FahrenbergZ24, AmraneBFZ25, AmraneBFF24, ClementEL26}.
%there is a version of the Kleene theorem~\cite{FahrenbergJSZ22}, a Myhill-Nerode theorem~\cite{FahrenbergZ23}, ....
Most relevant for our purpose is the work developed in~\cite{ClementEL26}, which proves a variant of Kamp's theorem for HDA.
In that paper, the authors define a temporal logic on pomsets that has the same expressive power as first order logic.
This is interesting for automatic verification of concurrent programs: model a program using an HDA; specify its desired behavior using LTL on pomsets; and algorithmically check that the program satisfies its intended specification.

Following McNaughton and Papert, our goal is to find two additional characterizations of the class of FO-definable HDA languages: (i) via an algebraic characterization akin to aperiodic monoids, and (ii) using a notion of \emph{counter-free HDA}.
Extending these classic results of automata theory to the higher-dimensional setting shows that FO and LTL are robust specification languages for model-checking HDA properties.
We leave open the question of finding a notion of star-free regular expression for HDA that captures the same class of languages.

\subparagraph{Pomset languages of HDA.}

While standard finite state automata recognize languages of words, HDA recognize languages of \emph{partially ordered multi-sets}, or \emph{pomsets}.
A word can be viewed as a special case of pomset, where the letters are totally ordered; such pomsets correspond to sequential executions.
But in general, two letters in a pomset need not be ordered: when two letters are incomparable, we think of them as two events occurring concurrently.
More precisely, events are ordered according to Lamport's happens-before partial order~\cite{Lamport86}. We say that event~$a$ \emph{precedes} event~$b$, written ${a \to b}$, if event~$a$ is terminated before $b$ starts.
Thus, two events are incomparable (a.k.a.\ \emph{concurrent}) when the two intervals during which they are executed overlap.

\begin{figure}[h]
\begin{tabular}{ccccc}
$
\left(
\begin{tikzpicture}[scale=1, baseline=2.5ex, every node/.style={transform shape}]
    \node (b) at (0,0) {$\ibullet b$};
    \node (d) at (1,0) {$\vphantom{\ibullet} d \ibullet$};
    \node (a) at (0.5,0.8) {$\vphantom{\ibullet} a \ibullet$};
    \path[precedence] (b) edge (d);
%    \path[event-order] (a) edge (b) (a) edge (d);
\end{tikzpicture}
\right)
$
& {\Large $\glue$} &
$
\left(
\begin{tikzpicture}[scale=1, baseline=2.5ex, every node/.style={transform shape}]
    \node (a) at (0,0.8) {$\ibullet a \vphantom{\ibullet}$};
    \node (c) at (1,0.8) {$ \vphantom{\ibullet}c\vphantom{\ibullet}$};
    \node (d) at (0.5,0) {$ \ibullet d\ibullet$};
    \path[precedence] (a) edge (c);
%    \path[event-order] (a) edge (d) (d) edge (c);
\end{tikzpicture}
\right)
$
& {\Large $=$} &
$
\begin{tikzpicture}[scale=1, baseline=2.5ex, every node/.style={transform shape}]
    \node (a) at (0,0.8) {$\vphantom{\ibullet}a$};
    \node (c) at (1,0.8) {$c\vphantom{\ibullet}$};
    \node (b) at (0,0) {$\ibullet b$};
    \node (d) at (1,0) {$d\ibullet$};
    \path[precedence] (b) edge (d) (a) edge (c) (b) edge (c);
%    \path[event-order] (a) edge (b) (a) edge (d) (c) edge (d);
\end{tikzpicture}
$
\\[0.8cm]
\scalebox{0.9}{
\begin{tikzpicture}[baseline=2.5ex]
    \def\hw{0.3}
    \def\fh{0.6}
    \def\sh{0.1}

    \coordinate (O) at (0,0);
    \coordinate[above = 1cm of O] (Up);
    \coordinate[right = 1.5cm of O] (Right);
    \coordinate[above = 1cm of Right] (Up-right);
    \draw[-] (O) -- (0,\sh);
    \draw[-] (0, \sh+\hw) -- (Up);
    \draw[-] (Right) -- (Up-right);

    \fill[pomset-1] (0.2,\fh) -- (1.51,\fh) -- (1.51,\fh+\hw) -- (0.2, \fh+\hw) --  cycle; 
    \fill[pomset-1] (1.53,\fh) -- (1.55,\fh) -- (1.55,\fh+\hw) -- (1.53, \fh+\hw) --  cycle; 
    \fill[pomset-1] (1.57,\fh) -- (1.59,\fh) -- (1.59,\fh+\hw) -- (1.57, \fh+\hw) --  cycle; 
    \fill[pomset-1] (1.61,\fh) -- (1.63,\fh) -- (1.63,\fh+\hw) -- (1.61, \fh+\hw) --  cycle;
    \draw[-]  (1.51,\fh+\hw) -- (0.2, \fh+\hw) -- (0.2,\fh) -- (1.51,\fh)  ;
    \node at (0.85,\fh+\hw*0.5) {$a$};

    \fill[pomset-2-light] (-0.01,\sh) -- (0.8,\sh) -- (0.8,\sh+\hw) -- (-0.01, \sh+\hw) --  cycle; 
    \fill[pomset-2-light] (-0.13,\sh) -- (-0.11,\sh) -- (-0.11,\sh+\hw) -- (-0.13, \sh+\hw) --  cycle; 
    \fill[pomset-2-light] (-0.09,\sh) -- (-0.07,\sh) -- (-0.07,\sh+\hw) -- (-0.09, \sh+\hw) --  cycle; 
    \fill[pomset-2-light] (-0.05,\sh) -- (-0.03,\sh) -- (-0.03,\sh+\hw) -- (-0.05, \sh+\hw) --  cycle; 
    \draw[-] (0,\sh) --(0.8, \sh) --  (0.8,\sh+\hw) -- (0,\sh+\hw);
    \node at (0.4,\sh+\hw*0.5) {$b$};

    \fill[pomset-3-light] (1,\sh) -- (1.51,\sh) -- (1.51,\sh+\hw) -- (1, \sh+\hw) --  cycle; 
    \fill[pomset-3-light] (1.53,\sh) -- (1.55,\sh) -- (1.55,\sh+\hw) -- (1.53, \sh+\hw) --  cycle; 
    \fill[pomset-3-light] (1.57,\sh) -- (1.59,\sh) -- (1.59,\sh+\hw) -- (1.57, \sh+\hw) --  cycle; 
    \fill[pomset-3-light] (1.61,\sh) -- (1.63,\sh) -- (1.63,\sh+\hw) -- (1.61, \sh+\hw) --  cycle; 
    \draw[-]  (1.51,\sh+\hw) -- (1, \sh+\hw) -- (1,\sh) -- (1.51,\sh)  ;
    \node at (1.25,\sh+\hw*0.5) {$d$};
\end{tikzpicture}
}
& {\Large $\glue$} &
\scalebox{0.9}{
\begin{tikzpicture}[baseline=2.5ex]
    \def\hw{0.3}
    \def\fh{0.6}
    \def\sh{0.1}

    \coordinate (O) at (0,0);
    \coordinate[above = 1cm of O] (Up);
    \coordinate[right = 1.5cm of O] (Right);
    \coordinate[above = 1cm of Right] (Up-right);
    \draw[-] (O) -- (Up);
    \draw[-] (Right) -- (Up-right);

    \fill[pomset-1] (-0.01,\fh) -- (0.5,\fh) -- (0.5,\fh+\hw) -- (-0.01, \fh+\hw) --  cycle; 
    \fill[pomset-1] (-0.13,\fh) -- (-0.11,\fh) -- (-0.11,\fh+\hw) -- (-0.13, \fh+\hw) --  cycle; 
    \fill[pomset-1] (-0.09,\fh) -- (-0.07,\fh) -- (-0.07,\fh+\hw) -- (-0.09, \fh+\hw) --  cycle; 
    \fill[pomset-1] (-0.05,\fh) -- (-0.03,\fh) -- (-0.03,\fh+\hw) -- (-0.05, \fh+\hw) --  cycle; 
    \draw[-] (-0.01,\fh+\hw) -- (0.5, \fh+\hw) -- (0.5,\fh) -- (-0.01,\fh);
    \node at (0.25,\fh+\hw*0.5) {$a$};

    \filldraw[pomset-4] (0.7,\fh) -- (1.3,\fh) -- (1.3,\fh+\hw) -- (0.7, \fh+\hw) --  cycle; 
    \node at (1,\fh+\hw*0.5) {$c$};

    \fill[pomset-3-light] (-0.01,\sh) -- (1.51,\sh) -- (1.51,\sh+\hw) -- (-0.01, \sh+\hw) --  cycle; 
    \fill[pomset-3-light] (1.53,\sh) -- (1.55,\sh) -- (1.55,\sh+\hw) -- (1.53, \sh+\hw) --  cycle; 
    \fill[pomset-3-light] (1.57,\sh) -- (1.59,\sh) -- (1.59,\sh+\hw) -- (1.57, \sh+\hw) --  cycle; 
    \fill[pomset-3-light] (1.61,\sh) -- (1.63,\sh) -- (1.63,\sh+\hw) -- (1.61, \sh+\hw) --  cycle; 
    \fill[pomset-3-light] (-0.13,\sh) -- (-0.11,\sh) -- (-0.11,\sh+\hw) -- (-0.13, \sh+\hw) --  cycle; 
    \fill[pomset-3-light] (-0.09,\sh) -- (-0.07,\sh) -- (-0.07,\sh+\hw) -- (-0.09, \sh+\hw) --  cycle; 
    \fill[pomset-3-light] (-0.05,\sh) -- (-0.03,\sh) -- (-0.03,\sh+\hw) -- (-0.05, \sh+\hw) --  cycle; 
    \draw[-] (1.51,\sh+\hw) -- (-0.01, \sh+\hw);
    \draw[-] (-0.01,\sh) -- (1.51,\sh);
    \node at (0.75,\sh+\hw*0.5) {$d$};
\end{tikzpicture}
}
& {\Large $=$} &
\scalebox{0.9}{
\begin{tikzpicture}[baseline=2.5ex]
    \def\hw{0.3}
    \def\fh{0.6}
    \def\sh{0.1}

    \coordinate (O) at (0,0);
    \coordinate[above = 1cm of O] (Up);
    \coordinate[right = 3cm of O] (Right);
    \coordinate[above = 1cm of Right] (Up-right);
    \draw[-] (O) -- (0,\sh);
    \draw[-] (0, \sh+\hw) -- (Up);
    \draw[-] (Right) -- (Up-right);

    \filldraw[pomset-1] (0.2,\fh) -- (2,\fh) -- (2,\fh+\hw) -- (0.2, \fh+\hw) --  cycle; 
    \node at (1.1,\fh+\hw*0.5) {$a$};

    \filldraw[pomset-4] (2.2,\fh) -- (2.8,\fh) -- (2.8,\fh+\hw) -- (2.2, \fh+\hw) --  cycle; 
    \node at (2.5,\fh+\hw*0.5) {$c$};

    \fill[pomset-2-light] (-0.01,\sh) -- (0.8,\sh) -- (0.8,\sh+\hw) -- (-0.01, \sh+\hw) --  cycle; 
    \fill[pomset-2-light] (-0.13,\sh) -- (-0.11,\sh) -- (-0.11,\sh+\hw) -- (-0.13, \sh+\hw) --  cycle; 
    \fill[pomset-2-light] (-0.09,\sh) -- (-0.07,\sh) -- (-0.07,\sh+\hw) -- (-0.09, \sh+\hw) --  cycle; 
    \fill[pomset-2-light] (-0.05,\sh) -- (-0.03,\sh) -- (-0.03,\sh+\hw) -- (-0.05, \sh+\hw) --  cycle; 
    \draw[-] (0,\sh) --(0.8, \sh) --  (0.8,\sh+\hw) -- (0,\sh+\hw);
    \node at (0.4,\sh+\hw*0.5) {$b$};

    \fill[pomset-3-light] (1,\sh) -- (3.01,\sh) -- (3.01,\sh+\hw) -- (1, \sh+\hw) --  cycle; 
    \fill[pomset-3-light] (3.03,\sh) -- (3.05,\sh) -- (3.05,\sh+\hw) -- (3.03, \sh+\hw) --  cycle; 
    \fill[pomset-3-light] (3.07,\sh) -- (3.09,\sh) -- (3.09,\sh+\hw) -- (3.07, \sh+\hw) --  cycle; 
    \fill[pomset-3-light] (3.11,\sh) -- (3.13,\sh) -- (3.13,\sh+\hw) -- (3.11, \sh+\hw) --  cycle; 
    \draw[-]  (3.01,\sh+\hw) -- (1, \sh+\hw) -- (1,\sh) -- (3.01,\sh)  ;
    \node at (2,\sh+\hw*0.5) {$d$};
\end{tikzpicture}
}
\end{tabular}
\caption{Gluing of pomsets and their interval representation}
\label{fig:gluing}
\end{figure}
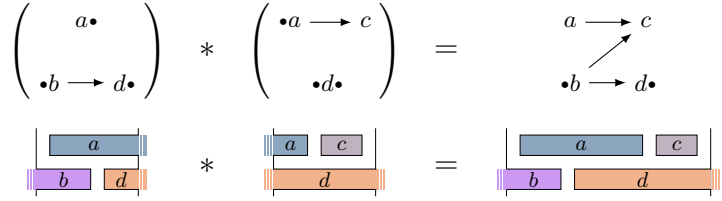

Additionally, pomsets are equipped with \emph{interfaces}.
The \emph{starting interface} is a set of events that are already started at the beginning of the computation. They are required to be minimal w.r.t.\ the precedence partial order.
Symmetrically, the \emph{terminating interface} is a set of events, maximal w.r.t.\ precedence, that are still ongoing at the end of the computation.
When two pomsets $P$ and $Q$ have matching interfaces (that is, the terminating interface of $P$ is the same as the starting interface of $Q$), we can \emph{glue} them together to obtain a new pomset $P \glue Q$.

Figure~\ref{fig:gluing} depicts the gluing of two pomsets, viewed as partial orders (top), and as intervals (bottom).
Events that belong to the starting (resp., terminating) interface are marked with a dot to the left (resp., right) of the letter.
Gluing plays the role of word concatenation.
However, since it is defined only for pomsets with matching interfaces, the underlying algebraic structure is that of a category rather than a monoid.
We write $\iIPoms$ for the category whose objects are interfaces, and morphisms are pomsets.

\subparagraph{Contributions.}

We first prove in \cref{sec:recognizability} the following algebraic characterization of HDA languages: a language of pomsets is regular if and only if it is the inverse image of a functor from $\iIPoms$ into a finite category (\Cref{thm:recognizable}).
The left-to-right implication follows the standard construction of the syntactic monoid associated with a language; in our setting, this is now a \emph{syntactic category}.
The converse direction requires a bit more work however.
While, for finite words, a morphism into a finite monoid can immediately be viewed as an automaton using the Cayley graph of the monoid; this is not the case for HDA.
What we get instead is what we call an $\iIPoms$-module. To turn it into an HDA, we rely on the canonical suffix presentation of a language of pomsets~\cite{FahrenbergZ24}.

Then, in \cref{sec:fo} we investigate the class of languages that are recognized by a functor into an aperiodic category.
A category is aperiodic when every endomorphism $x$ is such that $x^n = x^{n+1}$ for some~$n$.
We show in \cref{thm:aper-cat-fo} that the class of languages recognized aperiodic categories is also the class of FO-definable languages.
To prove this, we use the so-called ST-decomposition of a pomset.
This allows us to view a language of pomsets as a language of words, and relying on results from~\cite{ClementEL26}, we can transport the McNaughton-Papert theorem from words to pomsets.

Lastly in \cref{sec:aperiodicity}, we define a notion of counter-free HDA, and show that every language accepted by a counter-free HDA is aperiodic.
However, the converse direction remains an open question.
This is because the geometric structure of HDAs interacts with counter-freeness in a non-trivial way.
We explain why a naive attempt at solving the open question fails, and conclude with an example of a pomset language that exhibits some form of counting, despite being aperiodic.

%\subparagraph{Related work.}
%
%TODO

%\subparagraph{Plan of the paper.}
%
%TODO

%% file: prelim.tex
\section{Preliminaries}

In this section, we briefly recall the definitions about pomsets and HDA that we will need in the rest of the paper.
We encourage readers unfamiliar with pomset languages of HDA to read a more gentle introduction to those notions, e.g.~\cite{FahrenbergJSZ24}.
We also rely on a few notions from category theory, namely \emph{categories} and \emph{functors}. The necessary definitions can be found in the first chapter of any category theory textbook, e.g.~\cite{Awodey2010}.

\subparagraph*{Notations.}

For a category $\cC$, we write $\Ob{\cC}$ for the objects of $\cC$, and $\Mor{\cC}$ for the morphisms.
Composition is written in diagrammatic order, i.e., given morphisms $f : A \to B$ and $g : B \to C$, their composition is denoted $fg$ or $f \comp g$ rather than $g \circ f$.

\subsection{Higher Dimensional Automata (HDA)}

Fix a finite alphabet $\Sigma$.
An HDA is built from higher-dimensional cells.
Each cell is labeled with a finite list of elements of $\Sigma$, that we think of as the currently running events.
This totally ordered list of events is called a concurrency list, or \emph{conclist}.
\begin{definition}[Conclist]
A conclist $(U, \eventorder, \lambda)$ consists of a finite set $U$ equipped with a strict total order $\eventorder$, and a labeling function $\lambda : U \to \Sigma$.
\end{definition}
A conclist may have several occurrences of the same letter $a \in \Sigma$ running in parallel, in which case the total order $\eventorder$ is crucial to differentiate the corresponding events.
We write $\CList$ for the set of conclists.

\begin{definition}[HDA]
An HDA is a tuple $(X, \ev, \{\delta^0_{A,U}, \delta^1_{A,U} \mid A \subseteq U \in \CList \}, \bot, \top)$, where:
\begin{itemize}
\item $X$ is a finite set of cells, and $\ev : X \to \CList$ assigns a conclist to each cell. For a conclist~$U$, we write $X[U] = \{x \in X \mid \ev(x) = U \}$ for the set of cells of type~$U$.
\item For each $U \in \CList$ and $A \subseteq U$,  $\delta^0_{A,U} : X[U] \to X[U-A]$ is called the lower face map, and $\delta^1_{A,U} : X[U] \to X[U-A]$ the upper face map.
We often omit the subscript $U$ when it is clear from context.
The face maps must satisfy the precubical identities: for all $\mu, \nu \in \{0,1\}$ and $A \cap B = \emptyset$, we require $\delta^\mu_A \delta^\nu_B = \delta^\nu_B \delta^\mu_A$.
\item $\bot, \top \subseteq X$ are sets of initial and accepting cells, respectively. Often, $\bot \subseteq X[\emptyset]$.
\end{itemize}
\end{definition}

We often denote an HDA by its underlying set $X$.
The \emph{dimension} of an HDA is the maximal number of events active in a cell, $\dim(X) = \max\{\card{\ev(x)} \mid x \in X\}$.
A path in an HDA is a sequence of cells, where consecutive cells can be connected in two ways: either by starting some events, that is, moving from a lower face of~$x$ to the cell~$x$; or by terminating some events, that is, moving from a cell~$x$ to an upper face of~$x$.

\begin{definition}[Path]
A path in an HDA~$X$ is a sequence $\alpha = (x_0, \phi_1, x_1, \ldots, x_{n-1}, \phi_n, x_n)$,
where each $x_i \in X$, each $\phi_i \in \{ \upstep{A}, \downstep{A} \mid A \in \CList \}$, such that for all $1 \leq i \leq n$, either
\begin{itemize}
\item $\phi_i = \upstep{A}$ for some $A \subseteq \ev(x_i)$, and $x_{i-1} = \delta^0_A(x_i)$, or
\item $\phi_i = \downstep{A}$ for some $A \subseteq \ev(x_{i-1})$, and $\delta^1_A(x_{i-1}) = x_i$.
\end{itemize}
\end{definition}

The \emph{source} and \emph{target} of a path are $\src(\alpha) = x_0$ and $\tgt(\alpha) = x_n$, respectively.
A path is \emph{accepting} when $\src(\alpha) \in \bot$ and $\tgt(\alpha) \in \top$.
To define the language accepted by an HDA, we first need to introduce pomsets.

\begin{example}
A $2$-dimensional HDA is depicted in \cref{fig:hda-example}.
It consists of four $2$-dimensional cells, depicted as gray squares, twelve $1$-dimensional cells, depicted as edges, and nine $0$-dimensional cells, depicted as vertices.
Cells are labeled with the associated conclist. For example, the bottom-left square has two events $a$ and $b$ running in parallel, so its associated conclist is $U = \{a \eventorder b\}$, which we sometimes depict vertically as $\conclist{a\\b}$.
From this square, the lower face map $\delta^0_{\{a\},U}$ assigns the leftmost $b$-labeled edge, where event $a$ has not started yet; and the upper face map $\delta^1_{\{a\},U}$ assigns the middle $b$-labeled edge, where event $a$ is terminated. Similarly, the two other face maps $\delta^0_{\{b\},U}$ and $\delta^1_{\{b\},U}$ assign the two horizontal $a$-labeled boundaries of the $\conclist{a\\b}$-labeled square.
\end{example}

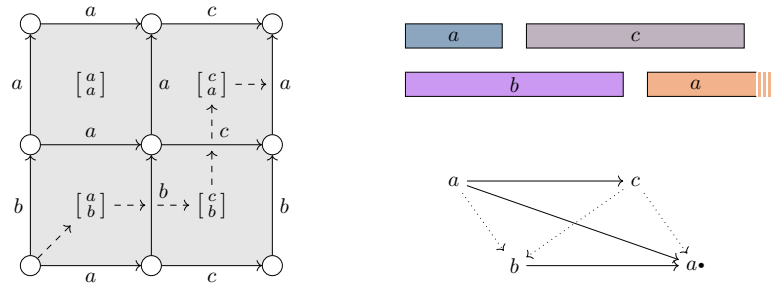
\begin{figure}[h]
\centering
\scalebox{0.8}{
\begin{tikzpicture}[x=1cm, y=1cm]
    % \path[use as bounding box] (0,-.7) -- (2,-.7) -- (2,2.6) -- (0,2.6) -- (0,-.7);
    \path[fill=black!10!white] (0,0) -- (4,0) -- (4,4) -- (0,4)-- (0,0);
    \node (x) at (1,1) {\Large$\conclist{a\\b}$};
    \node (y) at (3,1) {\Large$\conclist{c\\b}$};
    \node (w) at (1,3) {\Large$\conclist{a\\a}$};
    \node (z) at (3,3) {\Large$\conclist{c\\a}$};
    \node[state, minimum size=5pt, fill=white] (00) at (0,0) {};
    \node[] (phantom1) at (1.7,2) {};
    % \node[] (phantom1) at (1.8,2) {//};
    % \node[] (phantom2) at (0,1) {};
    \node[state, minimum size=5pt, fill=white] (01) at (0,2) {};
    \node[state, minimum size=5pt, fill=white] (02) at (0,4) {};
    \node[state, minimum size=5pt, fill=white] (10) at (2,0) {};
    \node[state, minimum size=5pt, fill=white] (11) at (2,2) {};
    \node[state, minimum size=5pt, fill=white] (12) at (2,4) {};
    \node[state, minimum size=5pt, fill=white] (20) at (4,0) {};
    \node[state, minimum size=5pt, fill=white] (21) at (4,2) {};
    \node[state, minimum size=5pt, fill=white] (22) at (4,4) {};
    \path[draw,->] (00) edge node[below]  {$a$} (10);
    \path[draw,->] (01) edge node[above] {$a$} (11);
    \path[draw,->] (10) edge node[above right] {$b$} (11);
    \path[draw,->] (00) edge node[left]  {$b$} (01);
    \path[draw,->] (10) edge node[below]  {$c$} (20);
    \path[draw,->] (11) edge node[above right] {$c$} (21);
    \path[draw,->] (12) edge node[above]  {$c$} (22);
    \path[draw,->] (01) edge node[left]  {$a$} (02);
    \path[draw,->] (02) edge node[above]  {$a$} (12);
    \path[draw,->] (11) edge node[right] {$a$} (12);
    \path[draw,->] (20) edge node[right]  {$b$} (21);
    \path[draw,->] (21) edge node[right]  {$a$} (22);
    \path[draw, dashed,->] (00) -- (x);
    \path[draw, dashed,->] (x) -- (1.9,1);
    \path[draw, dashed,->] (2.1,1) -- (y);
    \path[draw, dashed,->] (y) -- (3,1.9);
    \path[draw, dashed,->] (3,2.1) -- (z);
    \path[draw, dashed,->] (z) -- (3.9,3);
    \begin{scope}[xshift=6cm, yshift=2.8cm]
    	 \draw[pomset-1] (0.2,0.8) rectangle (1.8,1.2);
    	 \node at (1, 1) {$a$};
    	 \draw[fill=pomset-2-light] (0.2,0) rectangle (3.8,0.4);
    	 \node at (2, 0.2) {$b$};
    	 \draw[pomset-4] (2.2,0.8) rectangle (5.8,1.2);
    	 \node at (4, 1) {$c$};
    	 \fill[pomset-3-light] (6,0) rectangle (4.2,0.4);
    	 \fill[pomset-3-light] (6.04,0) rectangle (6.08,0.4);
    	 \fill[pomset-3-light] (6.12,0) rectangle (6.16,0.4);
    	 \fill[pomset-3-light] (6.20,0) rectangle (6.24,0.4);
    	 \draw[-] (6,0) -- (4.2,0) -- (4.2,0.4) -- (6,0.4);
    	 \node at (5, 0.2) {$a$};
    \end{scope}
    \begin{scope}[xshift=6cm, yshift=0cm]
    	 \node (a1) at (1, 1.4) {$a$};
    	 \node (b) at (2, 0) {$b$};
    	 \node (c) at (4, 1.4) {$c$};
    	 \node (a2) at (5, 0) {$a\ibullet$};
    	 \draw[->] (a1)--(c);
    	 \draw[->] (a1)--(a2);
    	 \draw[->] (b)--(a2);
    	 \draw[dotted,->] (a1) -- (b);
    	 \draw[dotted,->] (c) -- (b);
    	 \draw[dotted,->] (c) -- (a2);
    \end{scope}
  \end{tikzpicture}
  }
  \caption{An example of an HDA (left). Cells are labeled with associated conclists. The dashed arrows form a path, the pomset recognized by this path is illustrated at the right-hand side. The bullet at $a$ reflects the fact that the second occurrence of $a$ has not yet been terminated.}
  \label{fig:hda-example}
\end{figure}

\subsection{Pomset languages}
\label{sec:pomsets}

HDA recognize languages of partially ordered multi-sets, or \emph{pomsets}, a.k.a.\ labeled partial orders.
In fact, we require additional structure: an \emph{event order} $\eventorder$ to distinguish autoconcurrent events; source and target interfaces to allow for concatenation (or \emph{gluing}) of pomsets; and the partial order is always assumed to be an interval order.
For conciseness, we overload the word ``pomset'', always assuming the whole structure.

\begin{definition}[iiPomset]\label{def-pomset}
A pomset over $\Sigma$ 
is a tuple $(P, \precedence_{P}, \eventorder_{P}, \labelling_{P}, S_P, T_P)$ where:
\begin{itemize}
\item $(P, \precedence_{P})$ is a finite set equipped with a strict partial interval order called precedence.
\item $\eventorder_{P}$ is an acyclic relation on $P$ called the event order, such that for all $x \neq y \in P$, if $x \not\precedence_{P} y$ and $y \not\precedence_{P} x$, then either $x \eventorder_{P} y$ or $y \eventorder_{P} x$.
\item $\labelling_{P} \colon P \to \Sigma$ is a labeling function.
\item $S_P, T_P \subseteq P$ are the source (resp.\  target) interfaces, such that elements of $S_{P}$ (resp.\ $T_{P}$) are minimal (resp.\ maximal) elements w.r.t.~$\precedence_{P}$.
\end{itemize} 
%We write $x \parallel_{P} y$ when $x \not<_{P} y$ and $y \not<_P x$.
When there is no ambiguity, we denote $ \precedence_{P}, \eventorder_{P}$ and $\labelling_{P}$ as 
 $\precedence$, $\eventorder$ and $\labelling$.
\end{definition}

An example pomset is depicted in \cref{fig:hda-example}. In pictures, precedence $\precedence_{P}$ is represented using full arrows, and event order $\eventorder_{P}$ as dotted arrows, with transitive arrows omitted.
The source (resp.\ target) interfaces are represented as bullets before (resp.\ after) the label of the event.

A pomset $P$ is called \emph{discrete} when it has an empty precedence order.
In that case, the event order $\eventorder$ is a strict total order.
%Thus, we will write discrete pomsets as lists of events, between square brackets, where the event order is omitted and goes implicitly from top to bottom.
%For instance, $P = \pomset{\pinterface a \ibullet}{\ibullet b \pinterface}$ is a discrete pomset with two concurrent events $a \eventorder b$, where $S_{P} = \{b\}$ and $T_{P} = \{a\}$. 
Given a conclist $(U, \eventorder, \lambda) \in \CList$ and $A \subseteq U$, we define the following discrete pomsets:
\begin{itemize}
%\item a \emph{conclist} if $S_{P} = T_{P} = \emptyset$, 
\item $\id_U = (U, \emptyset, \eventorder, \lambda, U, U)$ is an \emph{identity},
\item $\starter{U}{A} = (U, \emptyset, \eventorder, \lambda, U-A, U)$ is a \emph{starter},
\item $\terminator{U}{A}\; = (U, \emptyset, \eventorder, \lambda, U, U-A)$ is a \emph{terminator}.
\end{itemize}
Intuitively, $\id_U$ is an elementary pomset where the events in $U$ are currently active in parallel. In $\starter{U}{A}$, we just started the events~$A$, so that the events~$U$ are now active.
In $\terminator{U}{A}$, the events~$U$ were active, and we terminate the events~$A$.
In small dimensions, we write discrete pomsets using a bracket notation, where the event order implicitly goes from top to bottom. For instance, $\conclist{\interface a \interface \\ \pinterface b \interface}$ denotes the starter $\starter{U}{\{b\}}$, where $U$ is the two element conclist $a \eventorder b$.

The source and target interfaces $S_P, T_P$, together with the event order $\eventorder$ and labeling map $\lambda$, can be viewed as conclists. % $(S_P, \eventorder, \lambda)$ and $(T_P, \eventorder, \lambda)$.
For $U, V \in \CList$, we write $\iIPoms(U, V)$ for the set of pomsets with source interface~$U$ and target interface~$V$ (up to isomorphism).
We define the \emph{gluing} operation on pomsets, $\glue : \iIPoms(U, V) \times \iIPoms(V, W) \to \iIPoms(U, W)$.
\begin{definition}[Gluing]
Let $P \in \iIPoms(U, V)$ and $Q \in \iIPoms(V, W)$. Let $f : T_P \to S_Q$ be the unique conclist isomorphism between the interfaces of~$P$ and~$Q$.
Then, we define the pomset $P \glue Q = (R, \precedence_R, \eventorder_R, \lambda_R, S_R, T_R)$, where
    \begin{align*}
    R &= (P \sqcup Q) /_{x \equiv f(x)} &
    \lambda_R &= \lambda_{P} \cup \lambda_{Q} \\
    \precedence_R &=\; \precedence_P \cup \precedence_Q \cup \; (P-T_P) \times (Q - S_Q) &
    S_R &= S_P \\
    \eventorder_R &=\; \eventorder_P \cup \eventorder_Q &
    T_R &= T_Q
    \end{align*}
\end{definition}
See \Cref{fig:gluing} for an example of gluing (event order implicitly goes from top to bottom). Gluing is associative and $\id_U$ is a neutral element, this data assembles into a category $\iIPoms$, whose objects are conclists, morphisms are pomsets, and composition is gluing.

The \emph{dimension} of a pomset~$P$ is the size of a maximal $<_P$-antichain.
We write $\iIPoms_{\leq k}$ for the subcategory of $\iIPoms$ consisting of pomsets of dimension at most~$k$.
Note that $\iIPoms_{\leq k}$ has a finite number of objects: they are conclists of length at most~$k$.

\begin{definition}[Language, subsumption, downward-closure]
A \emph{language} of pomsets is a set of morphisms $L \subseteq \iIPoms$.
For $P, Q \in \iIPoms(U, V)$, we say that $P$ \emph{is subsumed by} $Q$, written $P \subsumes Q$, when there is a bijection $f : P \to Q$ preserving the interfaces, labeling and event order, such that $f(x) \precedence_Q f(y)$ implies $x \precedence_P y$.
Intuitively, the pomset~$Q$ is ``more concurrent'' than~$P$.
A language $L \subseteq \iIPoms$ is \emph{downward-closed} when $Q \in L$ and $P \subsumes Q$ implies $P \in L$.
For a language $L$, we define $\downclosure{L} \;= \{ P \in \iIPoms \mid {P \subsumes Q} \text{ for some } Q \in L \}$, called the downward-closure of~$L$.
\end{definition}

Let~$X$ be an HDA, and~$\alpha$ a path in~$X$.
We define the pomset associated with $\alpha$, $\ev(\alpha) \in \iIPoms(\ev(\src(\alpha)), \ev(\tgt(\alpha)))$, by induction on the length of~$\alpha$:
\begin{itemize}
\item If $\alpha = (x_0)$, then $\ev(\alpha) = \id_{\ev(x_0)}$.
%\item If $\alpha = (\beta, \upstep{A}, x_n)$, then $\ev(\alpha) = \ev(\beta) \glue \starter{\ev(x_n)}{A}$.
%\item If $\alpha = (\beta, \downstep{A}, x_n)$, then $\ev(\alpha) = \ev(\beta) \glue \terminator{\ev(\tgt(\beta))}{A}$.
\item If $\alpha = (\delta^0_A(x), \upstep{A}, x)$, then $\ev(\alpha) = \starter{\ev(x)}{A}$.
\item If $\alpha = (x, \downstep{A}, \delta^1_A(x))$, then $\ev(\alpha) = \terminator{\ev(x)}{A}$.
\item If $\alpha = \alpha_1 \ast \ldots \ast \alpha_n$ is a concatenation, then $\ev(\alpha) = \ev(\alpha_1) \glue \ldots \glue \ev(\alpha_n)$.
\end{itemize}

\begin{definition}[HDA language] 
The language accepted by an HDA $X$ is
\[
	\Lang(X) = \{ \ev(\alpha) \mid \text{$\alpha$ is an accepting path in $X$}\}
\]
\end{definition}
A language of pomsets is \emph{regular} if it is accepted by a finite HDA.

\begin{remark}
An important property of HDA languages is that regular languages are always downward-closed~\cite{FahrenbergJSZ21}.
This is due to the geometric structure of HDA, where whenever a higher-dimensional cell exists, all of its faces must also exist.
Thus, they model asynchronous computation, where one can never force events to occur simultaneously.
Other formalisms, such as first-order logic or recognizability by finite categories, have no such restriction.
Therefore, when comparing expressivity of HDA with other formalisms, we will explicitly restrict to the class of downward-closed languages.
\end{remark}

\subparagraph*{Global assumption.}

Throughout the paper, we only consider languages of pomsets of bounded dimension, $\iIPoms_{\leq k}$. This is a crucial assumption, as it is not yet known how to properly model concurrent programs with an unbounded numbers of processes using HDA~\cite{BaronnerBH24}.
%
%We also restrict to languages of pomsets whose source interface is $\emptyset$, meaning that no processes are active at the start of computation.
%Unlike the previous assumption, this restriction is made solely for ease of presentation, as it prevents some technical details.
%We write $\OiIPoms$ for the set of pomsets whose source interface is $\emptyset$.
%%For HDA, this amounts to requiring that $\bot \subseteq X[\emptyset]$.
%
%Also, some of our results require the language to be downward-closed, while others do not. We explicitly indicate whenever this assumption is needed.

\subsection{First-Order logic on pomsets}

Another way to describe pomset languages is via logical characterizations.
MSO logic over pomsets was introduced in~\cite{AmraneBFF24} and shown to be as expressive as HDA, for downward-closed languages.
The fragment of First Order (FO) logic has been studied in \cite{ClementEL26}, together with an equivalent LTL-like logic.
Here, we focus on the class of FO-definable languages of pomsets.
FO formulas are generated by the following grammar, where $a\in\Sigma$ and $x,y$ are variables:
\[\varphi, \psi ::= a(x) \mid \fostarter(x) \mid \foterminator(x) \mid \neg\varphi \mid \varphi\wedge\psi \mid \exists x.\varphi \mid x < y \mid x \dasharrow y\]

The semantics is straightforward: variables range over events of the pomset; $<$ and $\dasharrow$ stand for precedence and event-order, respectively; and $a$, $\fostarter$, $\foterminator$ stand for the labeling, starting interface and terminating interface, respectively.

\begin{definition}[Semantics of FO over pomsets]
Consider an FO formula~$\phi$, a pomset~$P = (P, \precedence_{P}, \eventorder_{P}, \labelling_{P}, S_P, T_P)$, and a valuation $\nu$ mapping variables to events of~$P$.
Then, the satisfaction relation $P,\nu \models \phi$ is defined inductively as follows:
    \begin{alignat*}{4}
        &P, \nu \models a(x)  &\;&\text{ if } \labelling_P (\nu(x)) = a &\qquad\qquad
        &P, \nu \models \exists x.\phi &\;&\text{ if } \exists p \in P.\;\; P, \nu[x \mapsto p] \models \phi\\
        &P, \nu \models \fostarter(x) &&\text{ if } \nu(x) \in S_{P} &
        &P, \nu \models \foterminator(x)  &&\text{ if } \nu(x) \in T_{P} \\
        &P, \nu \models \neg \phi &&\text{ if } P, \nu \not\models \phi &
        &P, \nu \models \phi \wedge \psi &&\text{ if }  P, \nu \models \phi \text{ and } P, \nu \models \psi \\
        &P, \nu \models x \precedence y  &&\text{ if } \nu(x) \precedence_{P} \nu(y) &
        &P, \nu \models x \eventorder y  &&\text{ if } \nu(x) \eventorder_{P} \nu(y)
    \end{alignat*}
When $\phi$ has no free variable, we write $P \models \phi$, and $\Lang(\phi) = \{P \in \iIPoms \mid P \models \phi \}$.
\end{definition}

%\subsection{Downward-closure does not preserve aperiodicity}
%\label{sec:aperiodicity:down-clos}

%Now that we have shown that an aperiodic module cannot trivially be translated into a counter-free HDA, let us formulate a
%higher-level remark about subsumption-downclosure and aperiodicity. All results over the expressivity of aperiodic models in
%pomsets hold without consideration for the downclosure of the language (see~\cite{ClementEL26} and \Cref{thm:aper-cat-fo}).
%We show here that downclosure does not preserve aperiodicity. This asserts that a whole class of aperiodic languages cannot be resognized by
%counter-free HDA because their downclosure is not aperiodic.
%
FO-definable languages may not be downward-closed.
In fact, the class of FO-definable languages is not closed under downward-closure, as shown in \cref{prop:FO-not-downward-closed}.
Thus, when comparing expressivity of HDA models  with logical (or algebraic) characterizations, we will focus on the class of downward-closed FO-definable languages.

\begin{proposition}
    \label{prop:FO-not-downward-closed}
    The language $\conclist{a\\a}^*$ is FO-definable, but $\downclosure{\conclist{a\\a}^*}$ is not.
\end{proposition}
\begin{proof}
    The language $\conclist{a\\a}^*$ is recognized by the following FO formula: $\forall x. a(x) \wedge \exists! y. x \parallel y$,
    where $\exists!$ means ``exists unique'', defined as usual, % and is defined by $\exists! x. \varphi := \exists x. (\varphi \wedge \forall y. \varphi[x \rightarrow y] \Rightarrow x=y)$
    and $x \parallel y := \neg(x<y \vee y<x \vee x = y)$.
    %It is therefore FO-definable.
    However, $\downclosure{\conclist{a\\a}^*} \cap\; a^* = (aa)^*$, which is not FO-definable~\cite{DiekertG08}. Since FO-definable languages are closed under intersection, $\downclosure{\conclist{a\\a}^*}$ cannot be defined in FO.
\end{proof}

%% file: recognizability.tex
\section{Recognizability for HDA languages}
\label{sec:recognizability}

In this section, we propose a notion of algebraic recognizability for pomset languages.
Over words, a language $L \subseteq \Sigma^*$ is recognized by a finite monoid $M$
if there exists a monoid morphism $\varphi: \Sigma^* \rightarrow M$ and a subset $J \subseteq M$ such that $L = \varphi^{-1}(J)$.
It is a well-known result that recognizable languages are exactly the regular languages~\cite{Eilenberg74}.
In higher dimension, the set $\iIPoms$ of all pomsets over a given alphabet has the structure of a category rather than a monoid.
The objects are interfaces, and the arrows are pomsets whose domain and codomain are the starting and terminating interfaces, respectively.
%
%The notion of recognizability can in fact be defined for a language in any small category $\cC$.
Thus, we get the following notion of recognizability for pomset languages:

%\begin{definition}
%    A language $L\subseteq \cC$ is \emph{recognizable} if there exists a finite category $\cD$, a language $K\subseteq \cD$ and a functor $F:\cC\to\cD$ such that $L=F^{-1}(K)$.
%\end{definition}

\begin{definition}
  A language $L\subseteq \iIPoms$ is \emph{recognizable} if there exists a finite category $\cD$, a set of morphisms $K\subseteq \Mor{\cD}$, and a functor $F:\iIPoms\to\cD$ such that $L = F^{-1}(K)$.
\end{definition}

%In the case of pomsets, $L \subseteq \iIPoms$ is recognizable if there exists a finite category $\cD$, a language $K \subseteq \cD$ and a functor $F:\cC\to\cD$ such that $L=F^{-1}(K)$.
%
The rest of the section is dedicated to proving the following \Cref{thm:recognizable}.
The first implication (\Cref{lem:reg-to-rec}) follows the standard construction of the syntactic monoid.
The second implication (consequence of \Cref{prop:MN-pres}) is less straightforward and requires the notion of $\iIPoms$-module.

\begin{theorem}
  \label{thm:recognizable}
  A pomset language $L \subseteq \iIPoms_{\leq k}$ is regular if and only if it is recognizable.
\end{theorem}

\subsection{From regularity to recognizability}

To show that any regular language is recognizable, we use a construction similar to the usual syntactic monoid.
Taking into account the interfaces, we describe in \Cref{def:synt-cat} the \emph{syntactic category} of a language.
Then, we prove in \Cref{lem:reg-to-rec} that if the language is regular, then its syntactic category is finite.
Since any language is recognized by its syntactic category, this concludes the first direction of \Cref{thm:recognizable}.

The definition of the syntactic category can be formulated for any category~$\cC$. We will then specialize to the case of $\cC = \iIPoms$.
Given a language $L \subseteq \Mor{\cC}$ and two morphisms $\alpha, \beta : U \to V$ in~$\cC$, we write $\alpha \sim \beta$ when $\forall \gamma,\delta.\; \gamma\alpha\delta\in L \Leftrightarrow \gamma\beta\delta\in L$.
This is an equivalence relation on $\cC(U,V)$, called the syntactic congruence.

\begin{definition}
  \label{def:synt-cat}
  The \emph{syntactic category} $\Syn(L)$ of a language $L\subseteq \Mor{\cC}$ has objects $\Ob{\Syn(L)}=\Ob{\cC}$, and morphisms
  \[
    \Syn(L)(U,V)
    =
    \cC(U,V)/{\sim}
    \]
    with the composition $[\alpha]\comp[\beta]=[\alpha\comp\beta]$.
\end{definition}

The syntactic category comes with an obvious projection $\pi_L:\cC\to\Syn(L)$ and a language
\[K = \{[\alpha]\mid \alpha\in L\}\subseteq \Syn(L)\]
such that $\pi_L^{-1}(K)=L$.

\begin{lemma}
  \label{lem:reg-to-rec}
  If $L\subseteq \iIPoms_{\leq k}$ is regular, then $Syn(L)$ is finite. As a consequence, $L$ is recognizable.
\end{lemma}
\begin{proof}
  Let $X$ be a finite HDA such that $\Lang(X)=L$. For $P\in\iIPoms(U,V)$, let
  \[t(P)=\{(x,y)\in X[U]\times X[V]\mid \text{there is a path }\alpha \text{ s.t.\ } \src(\alpha) = x, \tgt(\alpha) = y \text{ and } \ev(\alpha) = P \}.\]
  First, $t(P)=t(P')$ implies $P\sim P'$; second, there is only a finite number of possible values of~$t$, so a finite number of equivalence classes for $\sim$.
\end{proof}
  
\subsection{From recognizability to regularity}

Now that we have established that any regular language is recognizable, we aim to prove the converse. %, stated in \Cref{lem:rec-to-reg}.
For languages of words, translating a finite monoid $M$ together with a homomorphism $\varphi: \Sigma^* \to M$ into a finite automaton is fairly straightforward: the states of the automaton are the elements of
the monoid, and the transitions are given by $x \overset{a}{\rightarrow} y$ if $x \cdot \varphi(a) = y$.
For pomset languages, however, the situation is different.
Given a functor $F : \iIPoms \to \cD$ into a finite category~$\cD$, the analogue of the above construction does not yield an HDA directly.
Instead, what we get is what we call an $\iIPoms$-module, that is, a set equipped with a right action of $\iIPoms$.
This can then be turned into an HDA, but it requires some work.

\begin{definition}
  Let $\cC$ be a small category. A $\cC$-module is a set $M$ equipped with
  \begin{itemize}
  \item Functions $\src,\tgt:M\to\Ob{\cC}$. For $U,V\in\Ob{\cC}$ let $M(U,V)=\src^{-1}(U)\cap \tgt^{-1}(V)$.
  \item Functions $\mu_{U,V,W}:M(U,V)\times \cC(V,W)\to M(U,W)$. We write $\mu$ for the partial function $M\times\Mor{\cC}\to M$, and $m\cdot \alpha$ for $\mu(m,\alpha)$.
  \end{itemize}
  such that, for all $m \in M(U, V)$, $\alpha \in \cC(V, W)$, and $\beta \in \cC(W, T)$,
  \begin{itemize}	
  \item $m\cdot \id_V = m$, and
  \item $(m\cdot\alpha)\cdot\beta=m\cdot (\alpha\comp\beta)$.
  \end{itemize}
\end{definition}

\begin{definition}
	A homomorphism of $\cC$-modules is a function $\varphi: M \rightarrow M'$ such that for any $m \in M$ and $\alpha \in \Mor{\cC}$, $\src(\varphi(m)) = \src(m)$, $\tgt(\varphi(m)) = \tgt(m)$ and $\varphi(m \cdot \alpha) = \varphi(m) \cdot \alpha$.
\end{definition}

\begin{remark}
	In other words, a $\cC$-module $M$ is a family of functors $M(U,-):\cC\to \Set$ indexed by $U\in \Ob\cC$, and a homomorphism of $\cC$-modules $M\to M'$ is a family of natural transformations $M(U,-)\Rightarrow M'(U,-)$. 
\end{remark}

%Intuitively, a module is a simplified category where the only action is a right-hand composition with pomsets. 
%This simplified definition is actually sufficient to define a notion of recognizability by a module.

Notice that, for $\cC = \Sigma^*$ viewed as a category with one object, a $\Sigma^*$-module is a set $M$ together with a right monoid action of $\Sigma^*$, that is, a deterministic finite state automaton.
Here, we are interested in $\iIPoms$-modules instead.
For instance, the syntactic category $\Syn(L)$ is an $\iIPoms$-module, with $[P] \cdot Q = [P\glue Q]$.
More generally, any functor $F: \iIPoms_{\leq k} \to \cD$ can be viewed as an $\iIPoms_{\leq k}$-module, by taking $M(U,V) = \{ F(P) \mid P \in \iIPoms_{\leq k}(U, V) \}$ for all $U, V$, and then letting $m \cdot Q = m \comp F(Q)$.
Finally, $\iIPoms_{\leq k}$ itself is an $\iIPoms_{\leq k}$-module, with the right action $P \cdot Q = P \glue Q$.

\begin{definition}
	A language $L\subseteq \iIPoms_{\leq k}$ is recognized by an $\iIPoms_{\leq k}$-module $M$ if there exists a subset $J\subseteq M$ and a homomorphism of $\iIPoms_{\leq k}$-modules $\varphi:\iIPoms_{\leq k}\to M$ such that $L=\varphi^{-1}(J)$.
\end{definition}

%In the case of pomset languages, $L$ is recognized by an $\iIPoms$-module if there is a module homomorphism $\varphi$ and a set $J \subseteq M$ such that
%$L = \varphi^{-1}(J)$. 
Note that in the above definition, the morphism of $\iIPoms_{\leq k}$-modules $\varphi:\iIPoms_{\leq k}\to M$ is entirely determined by the image of the identity pomsets $\id_{U} \in \iIPoms_{\leq k}(U, U)$ for each interface~$U$.
These play the role of the ``initial states'', if we think of the $\iIPoms_{\leq k}$-module~$M$ as some sort of deterministic automaton.
When a language $L$ is recognized by an $\iIPoms_{\leq k}$-module, we call the tuple $(M, J, \varphi)$ a \emph{presentation} of $L$.

\begin{example}
\label{ex:presentation}
A presentation for the language $L=\downclosure{\conclist{a\\b}} +\ abc$ is given in \Cref{fig:pres-example-pres}. Each element of the module $M$
is represented by a node in the graph. The labels inside a node indicate the preimage of the element by~$\varphi$. The elements
of $J$ are represented by adding a double border to their nodes. The arrows
represent the action of some starter or terminator.
For example, from the element $\phi(b\interface)$, we can either terminate $b$ by acting on the right by $\interface b$, to reach $\phi(b)$; or start $a$ by acting on the right by $\conclist{\pinterface a\interface\\ \interface b \interface}$, to reach $\phi(\conclist{a\interface\\b\interface})$.
The elements drawn in gray are those from which no element of $J$ is accessible. Their preimage by $\phi$ is infinite. Some such elements have been omitted, e.g., $\phi(c\interface)$.
One can notice that a finite presentation is essentially equivalent to a deterministic complete ST-automaton (see \cref{sec:ST-automata}).
\begin{figure}[h]
    \centering
    \begin{tikzpicture}[scale=2, nodes={draw,align=center}, font=\footnotesize]
        \node at (0,0)                      (v00)   {$\varepsilon$};
        \node at (0,1)                      (v01)   {$b\interface$};
        \node at (0,2)                      (v02)   {$b$};
        \node at (1,0)                      (v10)   {$a\interface$};
        \node at (1,1)                      (v11)   {$\conclist{a\interface\\b\interface}$};
        \node at (1,2)                      (v12)   {$\conclist{a\interface\\b\pinterface},$\\$ba\interface$};
        \node at (2,0)     (v20)   {$a$};
        \node at (2,1)                      (v21)   {$\conclist{a\pinterface\\b\interface}$};
        \node [double] at (2,2)     (v22)   {$\conclist{a\\b}, ba,$\\$abc$};
        \node at (3,0)                      (v30)   {$ab\interface$};
        \node [double] at (3,1)            (v31)   {$ab$};
        \node at (3, 2)                   (v32)     {$abc\interface$};
        % \node [double] at (4,0)     (v40)   {$ab, \dots$};
        % \node at (4,0.9)                      (v41)   {$abb\interface,\dots$};
        \def\cq{gray!80}
        \node[color=\cq] at (5,0.2) (v40)	{$aa\interface, \ldots$}; 
        \node[color=\cq] at (5.2,1) (v41)	{$aa, bb, \ldots$}; 
        \node[color=\cq] at (5,1.8) (v42)	{$bb\interface, \ldots$}; 

        \draw[-latex, semithick]    (v00) -- node[left, draw=none] {$b\interface$}  (v01);
        \draw[-latex, semithick]    (v10) -- node[left, draw=none] {$\conclist{\interface a\interface\\ \pinterface b\interface}$} (v11);
        \draw[-latex, semithick]    (v01) -- node[left, draw=none] {$\interface b$} (v02);
        \draw[-latex, semithick]    (v11) -- node[left, draw=none] {$\conclist{\interface a\interface\\ \interface b\pinterface}$} (v12);
        \draw[-latex, semithick]    (v21) -- node[left, draw=none] {$\interface b$} (v22);
        \draw[-latex, semithick]    (v00) -- node[above, draw=none] {$a\interface$} (v10);
        \draw[-latex, semithick]    (v01) -- node[above, draw=none] {$\conclist{\pinterface a\interface\\ \interface b\interface}$} (v11);
        \draw[-latex, semithick]    (v02) -- node[above, draw=none] {$a\interface$} (v12);
        \draw[-latex, semithick]    (v10) -- node[above, draw=none] {$\interface a$} (v20);
        \draw[-latex, semithick]    (v11) -- node[above, draw=none] {$\conclist{\interface a\pinterface\\ \interface b\interface}$} (v21);
        \draw[-latex, semithick]    (v12) --  node[above, draw=none] {$\interface a$} (v22);
        \draw[-latex, semithick]    (v20) --  node[above, draw=none] {$b\interface $} (v30);
        \draw[-latex, semithick]    (v30) --  node[left, draw=none] {$\interface b$} (v31);
        \draw[-latex, semithick]    (v31) --  node[left, draw=none] {$c \interface$} (v32);
        \draw[-latex, semithick]    (v32) --  node[above, draw=none] {$\interface c$} (v22);
        \draw[bend left=30,-latex,color=\cq]    (v40) edge node[left, draw=none] {$\interface a$} (v41);
        \draw[bend left=30,-latex,color=\cq]    (v41) edge node[right, draw=none] {$a\interface$} (v40);
        \draw[bend left=30,-latex,color=\cq]    (v42) edge node[right, draw=none] {$\interface b$} (v41);
        \draw[bend left=30,-latex,color=\cq]    (v41) edge node[left, draw=none] {$b\interface$} (v42);
        \draw[-latex,color=\cq] (v31) edge node[above, draw=none] {$b\interface$} (v42);
        \draw[-latex, bend left=30, color=\cq] (v02) edge node[above, draw=none] {$b\interface$} (v42);
        \draw[-latex,color=\cq] (v31) edge node[above, draw=none] {$a\interface$} (v40);
        \draw[-latex, bend right=20,color=\cq] (v20) edge node[above right, draw=none] {$a\interface$} (v40);
        \draw[-latex, bend left=20,color=\cq] (v22) edge node[above, draw=none] {$b\interface$} (v42);
        \draw[-latex, bend left=10,color=\cq] (v22) edge node[above right=-2pt, draw=none, pos=0.7] {$a\interface$} (v40);
        % \draw[-latex, semithick, bend left]    (v41) edge (v40);
%        \draw[-latex, semithick] (v00) -- (v11);
%        \draw[-latex, semithick] (v11) -- (v22);
    \end{tikzpicture}
    \caption{A presentation for $L=\downclosure{\conclist{a\\b}} +\ abc$.}
    \label{fig:pres-example-pres}
\end{figure}
\end{example}

Importantly, notice that one may always assume that the morphism $\phi : \iIPoms_{\leq k} \to M$ is surjective.
If that is not the case, we can simply discard the elements of $M$ that are not in the image of $\phi$, so that $(\phi(\iIPoms_{\leq k}), J \cap \phi(\iIPoms_{\leq k}), \phi)$ is a surjective presentation of $L$.

\begin{lemma}
  \label{lem:cat-to-mod}
  Any language $L \subseteq \iIPoms_{\leq k}$ recognizable by a finite category can be recognized by a finite $\iIPoms_{\leq k}$-module.
\end{lemma}
\begin{proof}
%  Fix a finite category $\cD$, a full \ee{is this the right word ? I want an injective functor} functor $F: \cC\to\cD$ and a set $K \subseteq \cD$ such that $L = F^{-1}(K)$.
%  Now, let us show that $(\cD, K, F)$ is a presentation for $L$.
%  First, the morphisms of $\cD$ are a $\cC$-module. Indeed, given $F(x) \in \cD$, we can canonically fix $\src(F(x)) = \src(x)$ and $\tgt(F(x)) = \tgt(x)$.
%  Additionally, there is a canonical composition on the right: given $x,y \in \cC$ with $\tgt(x) = \src(y)$, $F(x) \cdot y = F(x \circ y)$, where $\circ$ is the composition in $\cC$.
Suppose $L$ is recognized by a functor ${F : \iIPoms_{\leq k} \to \cD}$ into a finite category, with $K \subseteq \Mor{\cD}$.
As we have seen, this gives rise to an $\iIPoms_{\leq k}$-module whose elements are the morphisms of $\cD$ in the image of~$F$.
Then $F$ can be seen as a morphism of $\iIPoms_{\leq k}$-modules, and since $F^{-1}(K) = L$, this concludes the proof.
\end{proof}

\paragraph*{Canonical suffix presentation}
%Over word languages, the Myhill-Nerode theorem ensures that there exists a minimal deterministic automaton whose states corresponds to all suffix languages $w \backslash L = \{v \mid m\cdot v \in L\}$.
%Modules over $\iIPoms$ being structurally similar to deterministic automaton, there exists a similar \emph{Myhill-Nerode presentation} for pomset languages, as explicited by \Cref{prop:MN-pres}.
The reason for introducing the notion of $\iIPoms_{\leq k}$-module is that every finite presentation of a pomset language can be turned into a \emph{canonical suffix presentation}.
Given a language $L \subseteq \iIPoms_{\leq k}$ and pomset $P$, the \emph{prefix quotient} $P \backslash L$ is defined by $P \backslash L = \{Q \in \iIPoms_{\leq k} \mid P * Q \in L\}$.
Since pomset languages have a Myhill-Nerode theorem~\cite{FahrenbergZ24}, a pomset language $L$ is regular if and only if the set $\suff(L) = \{P \backslash L \mid P \in \iIPoms_{\leq k}\}$ is finite.
While the set $\suff(L)$ does not have the structure of a category, it can be equipped with the structure of an $\iIPoms_{\leq k}$-module.
This module can then be turned into an HDA using the construction from~\cite{FahrenbergZ24}.

\begin{proposition}
  \label{prop:MN-pres}
  Let $L \subseteq \iIPoms_{\leq k}$ be a language.
  Then, the presentation $(\suff(L), J, \varphi)$, called the canonical suffix presentation of~$L$, defined by:
  \begin{itemize}
    \item %For $P \backslash L \in \suff(L)$,
    $\src(P\backslash L) = S_P$ and $\tgt(P\backslash L) = T_P$,
    \item %For $P\backslash L \in \suff(L)_U$\ee{notation = $\tgt^{-1}(U)$} and $Q \in {}_U\iIPoms_{\leq k}$,
    $P\backslash L \cdot Q = (P\glue Q)\backslash L$,
    \item $J = \{P \backslash L \mid P \in L\}$,
    \item $\varphi(P) = P \backslash L$,
  \end{itemize}
  recognizes $L$. Additionally, for any presentation $(M, K, \psi)$ recognizing $L$ where $\psi$ is surjective, there exists a module homomorphism $f : M \to \suff(L)$ such that $(f(M), f(K), f \circ \psi)$ is
  isomorphic to the canonical suffix presentation $(\suff(L), J, \varphi)$ of $L$.
\end{proposition}
\begin{proof}
  First, $(\suff(L), J, \varphi)$ is well-defined since $P\backslash L \cdot \id_{T_P} = (P\glue\id_{T_P}) \backslash L = P \backslash L$
  and $(P \backslash L \cdot Q) \cdot R = (P\glue Q\glue R)\backslash L = P \backslash L \cdot (Q\glue R)$. Then, it is clear from the definitions that $L = \varphi^{-1}(J)$.
  Now, let us prove that any presentation $(M, K, \psi)$ admits a homomorphism $f$ such that $(f(M), f(K), f\circ\psi)=(\suff(L),J,\varphi)$.

  Fix, for any $m \in M$, $f(m) = \psi^{-1}(m) \backslash L$. This is well-defined since if $\psi(P) = \psi(Q) = m$, then $P \backslash L = \{R \in \iIPoms_{\leq k} \mid m \cdot R \in L\} = Q \backslash L$. 
  First, $f$ is a module homomorphism.
	Observe that $\src(f(\psi(P))) = \src(\psi^{-1}(\psi(P))\backslash L) = \src(P\backslash L) = S_P = \src(\psi(P))$ as $\psi$ is itself a homomorphism. The same goes for $\tgt(f(\psi(P)))$. Now, for any $Q \in \iIPoms_{\leq k}$,
	$f(\psi(P))Q = (\psi^{-1}(\psi(P)) \backslash L)Q = (P \backslash L)Q = PQ \backslash L = \psi^{-1}(\psi(PQ)) \backslash L = f(\psi(PQ)) = f(\psi(P)Q)$.

	It is immediate that $f \circ \psi(P) = \psi^{-1} (\psi(P)) \backslash L = P \backslash L$. Hence, $f(M) = \suff(L)$ and $f\circ \psi = \varphi$.
  Now, $f(K) = \{\psi^{-1}(\psi(P)) \backslash L \mid \psi(P) \in K\} = \{P\backslash L \mid P \in L\} = J$.
\end{proof}

\begin{example}
The presentation from \cref{ex:presentation} is isomorphic to the canonical prefix presentation of $L=\downclosure{\conclist{a\\b}} +\ abc$.
Notice, for instance, that $ab\backslash L = \{\epsilon,c\}$ whereas $ba\backslash L = \{\epsilon\}$.
Therefore, any presentation of $L$ must have two distinct elements $\phi(ab) \neq \phi(ba)$.
\end{example}

\Cref{prop:MN-pres} suffices to conclude that recognizable languages of finite dimension are regular.
Indeed, if we show that the set $\suff(L)$ is finite, then we can conclude by the Myhill-Nerode theorem for HDA~\cite{FahrenbergZ24} that $L$ is regular.
If $L$ is recognizable, it is recognized by a finite presentation $(M, K, \psi)$ where $\psi$ is surjective.
Then, \Cref{prop:MN-pres} says that there is a module homomorphism $f : M \to \suff(L)$ such that $f(M) = \suff(L)$.
Since $M$ is finite, this implies that $\suff(L)$ is also finite.

The proof of the Myhill-Nerode theorem from~\cite{FahrenbergZ24} essentially shows how to construct an HDA from the canonical suffix presentation $(\suff(L), J, \varphi)$.
In \Cref{sec:module-HDA-direct}, we adapt this proof to directly turn any finite module recognizing~$L$ into an HDA accepting the same language.

\subsection{Direct construction from modules to HDA}
\label{sec:module-HDA-direct}

%\jl[inline]{Say that only in this section we assume empty starting interfaces?}

In this section, we give another, more direct proof that recognizable languages are regular.
The construction of the HDA from an $\iIPoms_{\leq k}$-module is a generalized version of the one in~\cite{FahrenbergZ24}.
In this section only, we restrict to languages of pomsets whose source interface is $\emptyset$, meaning that no processes are active at the start of computation.
This restriction is made solely for ease of presentation, as it prevents some technicalities (see~\cite{FahrenbergZ24} for details).
We write $\OiIPoms$ for the set of pomsets whose source interface is $\emptyset$.

%\begin{lemma}
%  \label{lem:rec-to-reg}
%  If $L \subseteq \iIPoms_{\leq k}$ is recognizable, then it is regular.
%\end{lemma}

%The proof of \Cref{lem:rec-to-reg} relies on the fact that an $\iIPoms_{\leq k}$-module is structurally similar to an HDA:
Given a finite presentation $(M, J, \varphi)$, the $U$-cells of the resulting HDA are elements of $M(\emptyset, U)$. Face maps are then
added in the same way as in \cite{FahrenbergZ24}. However, for the down faces to be well-defined, we need an additional condition called coherence.
In the following, for~$P$ a pomset and $A \subseteq P$ a subset of events of~$P$, we write $P-A$ for the subpomset induced by the events of~$P$ not belonging to~$A$.

\begin{definition}
  A presentation $(M, J, \varphi)$ is \emph{coherent} if
	for all $U$ and all $P,Q\in \iIPoms(\emptyset,U)$:
  \[
    \varphi(P)=\varphi(Q) \implies \text{for all } A\subseteq U.\;\; \varphi(P-A)=\varphi(Q-A),
	\]
\end{definition}
  
%If a presentation is incoherent, with $\varphi(P-A) \neq \varphi(Q-A)$ even though $\varphi(P) = \varphi(Q)$, then the down face $d^0_A([P])$ is not unique.

\begin{lemma}
  Any recognizable language has a coherent presentation.
\end{lemma}
\begin{proof}
  Fix a presentation $(M, J, \varphi)$ recognizing a language $L$.
  Given $P, Q \in \iIPoms_{\leq k}$, we write $P \sim Q$ if $S_P = S_Q$, $T_P = T_Q$ and, for any $A \subseteq T_P$,
	$\varphi(P - A) = \varphi(Q - A)$. Then, $\sim$ is an equivalence relation with a finite number of classes, since $M$ is finite and
	the dimension of $L$ is bounded. Moreover, $\iIPoms_{\leq k}/{\sim}$ is an $\iIPoms_{\leq k}$-module. We now prove that together with $\psi: P \mapsto [P]$ and $K = \varphi^{-1}(J)$,
	it defines a coherent finite presentation for~$L$.

	First, ${}_{\emptyset}\iIPoms_{\leq k} \xrightarrow{\psi} \iIPoms_{\leq k}/{\sim} \supseteq K$ is coherent.
	If $\psi(P) = \psi(Q)$, then by definition $P \sim Q$, which implies that for any $A \subseteq T_P = T_Q$,
	$P - A \sim Q - A$.
	Therefore, $\psi(P - A) = \psi(Q - A)$.
	Finally, it recognizes $L$: $\psi(P) \in K$ if and only if there exists $Q \sim P$ such that $\varphi(Q) \in J$. Since $\varphi(Q) = \varphi(P)$,
	this is equivalent to $P \in L$.
\end{proof}

We can now construct an HDA from a coherent presentation.
Observe that if $(M,J,\phi)$ is the canonical suffix presentation of $L$, the resulting HDA is the Myhill-Nerode HDA of~\cite{FahrenbergZ24}.
\begin{lemma}
  \label{lem:pres-to-hda}
  Let $(M,J,\varphi)$ be a coherent presentation of a downward-closed language $L$, where $\phi$ is surjective.
  %Define a relation $\sim$ on ${}_\emptyset\iIPoms_U$ by $P\sim Q$ if $\varphi(P)=\varphi(Q)$.
  Then $X$ defined by
  \begin{mathpar}
  X[U]=M(\emptyset, U),
  \and 
  \delta^0_A(\phi(P))=\phi(P-A),
  \and
  \delta^1_B(m)=m \cdot \terminator{U}{B},
  \\
  \bot_X=\{\phi(\id_\emptyset)\},
  \and
  \top_X=J,
  \end{mathpar}
  is a well-defined HDA such that $\Lang(X)=L$.
\end{lemma}
\begin{proof}
  First, note that since $\phi$ is surjective, $X[U] = M(\emptyset, U) = \phi(\iIPoms(\emptyset, U))$.
  The map $\delta^0_A$ is well-defined: if $\phi(P) = \phi(Q)$, then $\phi(P-A) = \phi(Q-A)$ because the presentation is coherent.
  The precubical identities are straightforward to check. For example, if $m = \phi(P) \in M(\emptyset, U)$, and $A,B \subseteq U$ are disjoint, $\delta^0_A(\delta^1_B(m)) = \phi((P \glue \terminator{U}{B}) - A) = {\phi((P-A) \glue \terminator{(U-A)}{B})} = \delta^1_B(\delta^0_A(m))$.
  So $X$ is a well-defined HDA.

  To prove that the language accepted by $X$ is $L$, we first prove the following claim: for any path $\alpha$ in $X$, $\tgt(\alpha) \backslash L \subseteq \ev(\alpha) \backslash L$.
  Note that $\tgt(\alpha) \in M$, and $m\backslash L$ is defined by $m\backslash L = \{P \in \iIPoms_{\leq k} \mid m \cdot P \in J\}$.
  We prove this by induction over $\alpha$.
	\begin{itemize}
		\item If $\alpha = (\phi(\id_\emptyset))$, then $\tgt(\alpha) \backslash L = L = \ev(\alpha) \backslash L$.
		\item If $\alpha = (\beta, \downstep{B}, m)$, then $m = \delta^1_B(\tgt(\beta)) = \tgt(\beta) \cdot \terminator{U}{B}$, where $\tgt(\beta) \in M(\emptyset, U)$.
        Let $R \in \tgt(\alpha)\backslash L$, that is, $m \cdot R \in J$.
        Since $m \cdot R = (\tgt(\beta) \cdot \terminator{U}{B}) \cdot R = \tgt(\beta) \cdot (\terminator{U}{B} \glue R)$, this yields $\terminator{U}{B} \glue R \in \tgt(\beta) \backslash L$.
		By induction hypothesis, we obtain $\terminator{U}{B} * R \in \ev(\beta) \backslash L$, which is in turn equivalent to
		$R \in (\ev(\beta) * \terminator{U}{B}) \backslash L = \ev(\alpha) \backslash L$.
		\item If $\alpha = (\beta, \upstep{A}, m)$, let $P$ be such that $\phi(P) = m$, and $U = T_P$.
		Let $R \in \tgt(\alpha)\backslash L$, which means that $m \cdot R = \phi(P \glue R) \in J$, that is, $P \glue R \in L$.
 		One can check (see \cite{FahrenbergZ24} for a proof) that $(P - A) \glue \starter{U}{A} \subsumes P$, and since $L$ is downward-closed, $(P - A) \glue \starter{U}{A} \glue R \in L$.
		Since $\tgt(\beta) = \delta^0_A(m) = \phi(P-A)$, this means that $\starter{U}{A} \glue R \in \tgt(\beta) \backslash L$.
		By induction hypothesis, we obtain
		$\starter{U}{A} \glue R \in \ev(\beta) \backslash L$. Then, $R \in (\ev(\beta) \glue \starter{U}{A}) \backslash L = \ev(\alpha) \backslash L$.
	\end{itemize}
  With this property, if $\alpha$ is an accepting path, $\tgt(\alpha) \in \top_X = J$, so $\id_U \in \tgt(\alpha)\backslash L$, which implies $\id_U \in \ev(\alpha) \backslash L$, i.e.\ $\ev(\alpha) \in L$.
  We have shown $\Lang(X) \subseteq L$.

  For the converse, we show that for any pomset $P \in {}_\emptyset\iIPoms_{\leq k}$, there is a path $\alpha$ in~$X$ starting in $\bot_X$ and ending in $\phi(P)$, with $\ev(\alpha) = P$.
  This implies that $L \subseteq \Lang(X)$, since for any $P \in L$, $\phi(P) \in J$ is an accepting cell.
  We proceed by induction over the length of an ST-decomposition of $P = P_1 \glue \cdots \glue P_n$ (see \cref{sec:ST-automata} for a reminder on ST-decompositions).
  If $n = 0$, then $P = \id_\emptyset$ and there is a path from $\bot_X$ to $\bot_X$ labeled by $\id_\emptyset$.
  Now, assume we have built a path $\beta$ from $\bot_X$ to $\phi(P_1 \glue \cdots \glue P_{n-1})$ with $\ev(\beta) = P_1 \glue \cdots \glue P_{n-1}$.
  There are two cases. If $P_n = \starter{U}{A}$ is a starter, then we take the path $\alpha = (\beta, \upstep{A}, \phi(P))$.
  This is a correct path because $\delta^0_A(\phi(P)) = \phi(P-A) = \tgt(\beta)$, because $P - A = (P_1 \glue \cdots \glue P_{n-1} \glue \starter{U}{A}) - A = P_1 \glue \cdots \glue P_{n-1}$.
  Similarly, if $P_n = \terminator{U}{A}$ is a terminator, we take $\alpha = (\beta, \downstep{A}, \phi(P))$ which is a correct path because $\delta^1_A(\tgt(\beta)) = \tgt(\beta) \cdot \terminator{U}{A} = \phi(P_1 \glue \cdots \glue P_{n-1} \glue \terminator{U}{A}) = \phi(P)$.
\end{proof}

\begin{example}
  Consider again the presentation of the language $L=\downclosure{\conclist{a\\b}} +\ abc$, depicted in \Cref{fig:pres-example-pres}.
  The corresponding HDA, constructed by \Cref{lem:pres-to-hda}, is depicted in \Cref{fig:pres-example-hda}.
  Notice that the vertical $b$-transition on the right side of the square does not exist in the original presentation.
  It corresponds to the element $\phi(\conclist{a\pinterface\\b\interface})$ of the module.
  The definition of HDA requires that this $1$-dimensional cell must have a lower face map, which we defined as $\delta^0_{\{b\}}(\phi(\conclist{a\pinterface\\b\interface})) = \phi(a)$.
  As a consequence, the resulting HDA is not deterministic, even though
  presentations, by construction, are. This is not
  surprising, since it was established in~\cite{FahrenbergZ24} that there is no deterministic HDA
  for this language $L$.
\end{example}

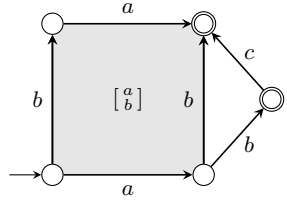
\begin{figure}[h]
    \centering
    \begin{tikzpicture}[scale=2, >=stealth, font=\footnotesize]
        % labels
        \node at (0.9,0.5)      {$b$};
        \node at (-0.1,0.5)     {$b$};
        \node at (0.5,-0.1)     {$a$};
        \node at (0.5,1.1)      {$a$};
        \node at (0.5,0.5)      {$\conclist{a\\b}$};
        \node at (1.3,0.2)     {$b$};
        \node at (1.3,0.8)     {$c$};
        % \node at (1.5,0.1)      {$b$};
        % \node at (2,0.55)       {$b$};
%            \node at (-0.12,-0.12)  {$\bot$};
%            \node at (1.12, -0.12)  {$\top$};
%            \node at (1.12,1.12)    {$\top$};
%            \node at (1.88,-0.12)   {$\top$};
        % square
        \fill[fill opacity = 0.1] (0,0) -- (0+1,0) --(0+1,0+1) --(0,0+1) -- cycle  ;
        % vertices
        \node[state, initial, initial text=, initial distance=6pt, fill=white, inner sep=0pt, minimum size=8pt] (v00) at (0,0) {};
        \node[state, fill=white, inner sep=0pt, minimum size=8pt] (v01) at (0,1) {};
        \node[state, fill=white, inner sep=0pt, minimum size=8pt] (v10) at (1,0) {};
        \node[state, accepting, fill=white, inner sep=0pt, minimum size=8pt] (v11) at (1,1) {};
        \node[state, accepting, fill=white, inner sep=0pt, minimum size=8pt] (v20) at (1.45,.5) {};
        % \node[state, accepting, fill=white, inner sep=0pt, minimum size=8pt] (v20) at (2,0) {};
        % edges
        \foreach \x in {0,1}
        {
            \draw[->, semithick] (v\x0) -- (v\x1);
        }
        \draw[->, semithick] (v00) -- (v01);
        \draw[->, semithick] (v10) -- (v11);
        \draw[->, semithick] (v00) -- (v10) ;
        % \draw[->, semithick] (v10) -- (v20) ;
        \draw[->, semithick] (v01) -- (v11) ;
        \draw[->, semithick] (v10) -- (v20) ;
        \draw[->, semithick] (v20) -- (v11) ;
        % \draw[->, very thick, bend left, red]  (v11) edge (v10) ;
        % \draw[->, semithick, loop above] (v20) edge (v20) ;
    \end{tikzpicture}
    \caption{The HDA constructed from the presentation of \Cref{ex:presentation} (co-accessible cells only). }
    \label{fig:pres-example-hda}
\end{figure}

%% file: firstorder.tex
\section{Aperiodic categories and First-Order logic}
\label{sec:fo}

For languages of finite words, algebraic characterizations have led to the study of several strict subclasses of regular languages, defined by the properties of their recognizing monoids.
Perhaps the most prominent is the class of languages recognized by aperiodic monoids.
This class can equivalently be characterized using star-free regular expressions, counter-free automata and first-order logic~\cite{Schutzenberger65,McNaughton71}.
Intuitively, these languages can be recognized without needing to count letters modulo a constant $k > 1$.
For instance, $a^*$ is star-free, whereas $(aa)^*$ is not, because in the latter requires counting the number of $a$'s modulo~2.

In this section, we extend this class of languages from words to pomsets.
Recall that a finite monoid $M$ is \emph{aperiodic} if there exists $n \in \mathbb{N}$ such that, for every $x\in M$, $x^{n+1} = x^n$.
Extending this definition to categories and modules is straightforward:
we require the same property for every endomorphism $x \in \cC(U,U)$.
For the particular case of pomsets, this means that we only iterate pomsets with the same source and target interfaces.

\begin{definition}
    A finite category $\cC$ is \emph{aperiodic} if there exists $n \in \mathbb{N}$ such that, for every object $U$ and endomorphism $x \in \cC(U,U)$, $x^{n+1} = x^n$.
\end{definition}

\begin{definition}
    A $\cC$-module $M$ is \emph{counter-free} if there exists $n \in \mathbb{N}$ such that, for every objects $U,V$ and for every $m \in M(V,U)$ and $c \in \cC(U,U)$, $m\cdot x^{n+1} = m \cdot x^n$.
\end{definition}

For modules, we prefer to use the term \emph{counter-free}, traditionally used for automata, because we consider them to be closer to automata than to categories. % since they are equipped only with an action from $\cC$ instead of an internal composition like monoids and categories are.
In this section, we aim to prove that aperiodic categories, counter-free modules, and FO logic recognize the same pomset languages.
Expressivity of counter-free HDA will be explored in \Cref{sec:aperiodicity}.

\subsection{ST-sequences and ST-automata}
\label{sec:ST-automata}

First, let us introduce a technical tool called the Starter-Terminator (ST) decomposition of a pomset.
In previous work, ST-decompositions have been used to lift various results from words to pomsets, such as the Büchi-Elgot-Trakhtenbrot theorem~\cite{AmraneBFF24} and Kamp theorem~\cite{ClementEL26}.
The key idea is that any pomset of bounded dimension can be represented as a finite word over a finite alphabet, called an ST-sequence.
%
% A pomset $(P,<_P, \eventorder_P, \lambda_P, S_P, T_P)$ is said to be \emph{discrete} if $<_P$ is empty (\emph{e.g.}, $\conclist{a\pinterface\\a\interface}$ is discrete).
% \begin{definition}
%     A discrete pomset $(P,<_P, \eventorder_P, \lambda_P, S_P, T_P)$ is a \emph{starter} if $P = T_P$. It is a \emph{terminator} if $P = S_T$.
% \end{definition}
% For instance, $P = \conclist{\interface a \interface \\ \interface b \pinterface}$ is a terminator, because the $b$ that was running in the starting
% interface is terminated in $P$, and only the $a$ continues running. In a terminator, the whole pomset is the starting interface, hence the $P = T_P$ in the definition. An
% event in $P - T_P$ is started by $P$.
% If a discrete pomset is both a starter and a terminator, we call it an \emph{identity}, since it does not start nor terminate any event: $S_P = P = T_P$.
%
Recall from \cref{sec:pomsets} that starters and terminators are special elementary pomsets consisting of concurrent events, some of which may be started or terminated.
We denote by $\Omega$ the set of all starters and terminators (including identities), and $\Omega_{\leq k}$ those of dimension at most~$k$. Note that $\Omega_{\leq k}$ is finite.
\begin{definition}
    A word $w=P_1 P_2 \cdots P_n \in \Omega^*$ is an \emph{ST-sequence} if $T_{P_i} = S_{P_{i+1}}$ for any $i = 1, 2, \dots, n-1$.
    In such a case, we define $\glueST(w) = P_1 \glue P_2 \glue \cdots \glue P_n \in \iIPoms$.
\end{definition}
Although any ST-sequence defines a unique pomset, a given pomset can be decomposed into several ST-sequences.
However, there exists a canonical ST-decomposition of a pomset.
\begin{proposition}[\cite{FahrenbergZ24}]
    For any pomset $P$, there exists a unique ST-sequence $w$ such that $\glueST(w) = P$ and $w$ alternates between non-identity starters and terminators.
\end{proposition}
\begin{figure}[h]
    \centering
	\scalebox{0.8}{
    \begin{tikzpicture}[x=1cm, y=1cm]
        \begin{scope}[xshift=3cm]
        	\draw[dotted]  (1,1.6)--(1,-0.8);
        	\draw[dotted]  (2,1.6)--(2,-0.8);
        	\draw[dotted]  (3,1.6)--(3,-0.8);
        	\draw[dotted]  (4,1.6)--(4,-0.8);
        	\draw[dotted]  (5,1.6)--(5,-0.8);
            \draw[pomset-1] (0.2,0.8) rectangle (1.8,1.2);
            \node at (1, 1) {$a$};
            \draw[fill=pomset-2-light] (0.2,0) rectangle (3.8,0.4);
            \node at (2, 0.2) {$b$};
            \draw[pomset-4] (2.2,0.8) rectangle (5.8,1.2);
            \node at (4, 1) {$c$};
            \fill[pomset-3-light] (6,0) rectangle (4.2,0.4);
            \fill[pomset-3-light] (6.04,0) rectangle (6.08,0.4);
            \fill[pomset-3-light] (6.12,0) rectangle (6.16,0.4);
            \fill[pomset-3-light] (6.20,0) rectangle (6.24,0.4);
            \draw[-] (6,0) -- (4.2,0) -- (4.2,0.4) -- (6,0.4);
            \node at (5, 0.2) {$a$};
            % ST-sequence
            \node at (0.5,-.5) {$\conclist{\colora\interface\\\colorb\interface}$}; 	
            \node at (1.5,-.5) {$\conclist{\interface \colora\pinterface\\\interface \colorb\interface}$};
            \node at (2.5,-.5) {$\conclist{\pinterface \colorc\interface\\\interface \colorb\interface}$};
            \node at (3.5,-.5) {$\conclist{\interface \colorc\interface\\\interface \colorb\pinterface}$};
            \node at (4.5,-.5) {$\conclist{\interface \colorc\interface\\\pinterface \colorabis\interface}$};
            \node at (5.5,-.5) {$\conclist{\interface \colorc\pinterface\\\interface \colorabis\interface}$};
        \end{scope}
        \begin{scope}[xshift=-3cm]
            \node (a1) at (1, 1.4) {$a$};
            \node (b) at (2, 0) {$b$};
            \node (c) at (4, 1.4) {$c$};
            \node (a2) at (5, 0) {$a\ibullet$};
            \draw[->] (a1)--(c);
            \draw[->] (a1)--(a2);
            \draw[->] (b)--(a2);
            \draw[dotted,->] (a1) -- (b);
            \draw[dotted,->] (c) -- (b);
            \draw[dotted,->] (c) -- (a2);
        \end{scope}
    \end{tikzpicture}
	}
    \caption{A pomset with its associated sparse ST-decomposition.}
    \label{fig:st-example}
\end{figure}
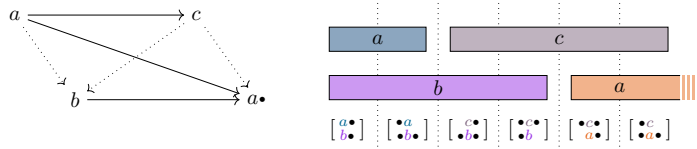
Such a decomposition is called \emph{sparse}.
As emphasized, non sparse ST-decompositions are not unique. For instance, the pomset depicted in \Cref{fig:st-example} has non-sparse ST-decomposition
$
    \conclist{\colora\interface\\~}
    \conclist{\interface\colora\interface\\\pinterface\colorb\interface}
    \conclist{\interface \colora\pinterface\\\interface \colorb\interface}
    \conclist{\pinterface \colorc\interface\\\interface \colorb\interface}
    \conclist{\interface \colorc\interface\\\interface \colorb\pinterface}
    \conclist{\interface \colorc\interface\\\pinterface \colorabis\interface}
    \conclist{\interface \colorc\pinterface\\\interface \colorabis\interface}
$, with two consecutive starters.
Given a language $K$ of ST-sequences, we denote by $\glueST(K) = \{\glueST(w) \mid w \in K\}$ the language of associated pomsets.
Conversely, given a language $L$ of pomsets of dimension $k$,
$\STseq(L) = \{w \in \Omega_{\leq k}^* \mid \glueST(w) \in L\}$ is the language of ST-decompositions representing pomsets of $L$.
Since ST-sequences are words over a finite alphabet, they can be accepted by finite state automata.
\begin{definition}[\cite{AmraneBCFZ24}]
    A \emph{$k$-dimensional ST-automaton} is a tuple $\mathcal A = (Q, I, F, \delta, \lambda)$ where
    $(Q, I, F, \delta)$ is a finite state automaton over the alphabet $\Omega_{\leq k}$ and $\lambda: Q \to \CList_{\leq k}$
    is a labeling function such that for any transition $q \overset{P}{\to} q'$, $\lambda(q) = S_P$ and $\lambda(q') = T_P$.
\end{definition}

\subsection{A McNaughton-Papert theorem for pomset languages}

\begin{theorem}
    \label{thm:aper-cat-fo}
    Let $L \subseteq \iIPoms_{\leq k}$ be a pomset language of bounded dimension $k$. The following propositions are equivalent:
    \begin{enumerate}
        \item\label{thm:aper-cat-fo:cat} $L$ is recognized by an aperiodic category,
        \item\label{thm:aper-cat-fo:mod} $L$ is recognized by a counter-free $\iIPoms$-module,
        \item\label{thm:aper-cat-fo:fo} $L$ is defined by an FO formula.
    \end{enumerate}
\end{theorem}

\begin{proof}
    (\ref{thm:aper-cat-fo:cat}) $\Rightarrow$ (\ref{thm:aper-cat-fo:mod}).
    Let $\cC$ be an aperiodic category recognizing $L$ with functor $F$. We prove that $\cC$, considered as an $\iIPoms_{\leq k}$-module as in the proof
    of \Cref{lem:cat-to-mod}, is counter-free.
    Fix $P, Q \in \iIPoms$. Then, $F(P) \cdot Q^{n+1} = F(P) \comp F(Q)^{n+1} = F(P) \comp F(Q)^n = F(P) \cdot Q^n$, where $\comp$ is the composition in
    $\cC$.
    So the associated module is counter-free.

    (\ref{thm:aper-cat-fo:mod}) $\Rightarrow$ (\ref{thm:aper-cat-fo:fo}).
    Let $(M,J,\varphi)$ be a counter-free presentation for $L$. Then, consider the ST-automaton $\mathcal{A} = (Q, I, F, \delta, \lambda)$
    with:
    \begin{mathpar}
	Q = M, \and 
	%I = \{\varphi(\varepsilon)\} \and 
	I = \{\varphi(\id_U) \mid U \in \CList_{\leq k}\}, \and
	F = J, \\
	\delta(m, P) = m \cdot \varphi(P), \and
	\lambda(m) = \tgt(m).
    \end{mathpar}
%    \begin{itemize}
%        \item $Q = M$% \{m \mid \exists U . m \in M(\emptyset, U)\}$
%        \item $I = \{\varphi(\varepsilon)\}$
%        \item $F = J \cap Q$
%        \item $\delta(m, P) = m \cdot \varphi(P)$
%    \end{itemize}
    First, observe that $\mathcal A$ is counter-free (see \cref{sec:aperiodicity} for a reminder of the notion of counter-freeness for word automata). Indeed, for any $m \in Q$ and $w \in \Omega_{\leq k}^*$, taking $i$ such that $m \cdot P^{i+1} = m \cdot P^i$, we have
    $\delta(m, w^{i+1}) = m \cdot \varphi(\glueST(w))^{i+1} = m \cdot \varphi(\glueST(w))^i = \delta(m, w^i)$.
    Further, $\Lang(\mathcal A) = \STseq(L)$, which implies that the word language $\STseq(L)$ is FO-definable~\cite{McNaughton71}.
    Hence, using a result from~\cite{ClementEL26}, $\glueST(\STseq(L)) = L$ can be recognized by an FO formula over pomsets.

    (\ref{thm:aper-cat-fo:fo}) $\Rightarrow$ (\ref{thm:aper-cat-fo:cat}).
    It is known from~\cite{ClementEL26} that any pomset language $L$ is defined by an FO formula if and only if
    $\STseq(L)$ is defined by an FO formula over ST-sequences. Since ST-sequences are words over finite alphabets,
    this implies that there exists an aperiodic monoid recognizing $\STseq(L)$~\cite{McNaughton71}.
    Fix $M$ to be such a monoid, with $J \subseteq M$ and $\varphi: \Omega^*_{\leq k} \to M$ a monoid morphism such that $\STseq(L) = \varphi^{-1}(J)$.
    Assume that $\varphi$ is surjective (if this is not the case, then we can replace $M$ by $\varphi(\Omega^*_{\leq k})$).
    We cannot recognize $L$ using $M$ and $\varphi$ because one pomset can be recognized by several ST-sequences.
    Let us define $\sim$ the equivalence relation over $M$ induced by $\varphi(v) \sim \varphi(w)$ whenever $\glueST(v) = \glueST(w)$.
    The relation $\sim$ is a congruence on $M$. Indeed, if $a,m,m'\in M$ and $m\sim m'$, 
	then there exist $v,w,w'\in \Omega^*_{\leq k}$ such that $\varphi(v)=a$, $\varphi(w)=m$, $\varphi(w')=m'$ and $\glueST(w)=\glueST(w')$. Then 
	\[
		\glueST(vw)=\glueST(v)*\glueST(w)=\glueST(v)*\glueST(w')=\glueST(vw'),
	\] 
	which implies $am=\varphi(v)\varphi(w)=\varphi(vw)\sim\varphi(vw')=\varphi(v)\varphi(w')=am'$. Similarly, $ma\sim m'a$. Thus, the quotient $M/\!\sim$ is a monoid. 
    Now, define $\psi:\iIPoms_{\leq k}\to M/\!\sim$ by the formula $\psi(P)=[\varphi(\glueST^{-1}(P))]$. This is valid since for $w,w'\in \Omega_{\leq k}^*$ such that $\glueST(w)=\glueST(w')=P$ we have $\varphi(w)\sim\varphi(w')$ and hence $[\varphi(w)]=[\varphi(w')]$. Since
    \[
    \psi(\glueST(v)*\glueST(w))=\psi(\glueST(vw))=[\varphi(vw)]=[\varphi(v)][\varphi(w)]=\psi(\glueST(v))\psi(\glueST(w)),
    \]
    $\psi$ is a homomorphism. Further, we choose $J' = \{[m] \mid m \in J\}$.
%    This is also well-defined: if $m\sim m'$, then $m=\varphi(w)$ and $m'=\varphi(w')$ for some $w,w'\in\Omega_{\leq k}^*$ such that $\glueST(w)=\glueST(w')$.
%    Then, $w\in\STseq(L)$ iff $w'\in\STseq(L)$, that is, $\varphi(w) \in J$ iff $\varphi(w') \in J$
%    and thus $[m] \in J'$ does not depend on the choice of the representative.
	%
    Notice that if $m\sim m'$, then $m \in J$ iff $m' \in J$. To see this, assume $m=\varphi(w)$ and $m'=\varphi(w')$ for some $w,w'\in\Omega_{\leq k}^*$ such that $\glueST(w)=\glueST(w')$.
    Then, $w\in\STseq(L)$ iff $w'\in\STseq(L)$, that is, $\varphi(w) \in J$ iff $\varphi(w') \in J$.
    From this, we deduce that $L = \psi^{-1}(J')$.
%    Indeed, let $P \in L$, with an ST-decomposition $w \in \STseq(L)$, i.e., such that $\phi(w) \in J$. Then $\psi(P) = [\phi(w)] \in J'$.
%    Conversely, assume $\psi(P) \in J'$, that is, $\phi(w) \sim m$ for some $m \in J$ and $w \in \Omega_{\leq k}$ with $\glueST(w) = P$.
	Indeed, $\psi(P) \in J'$ is equivalent to $\phi(w) \sim m$ for some $m \in J$ and $w \in \Omega_{\leq k}$ with $\glueST(w) = P$.
    Then $\phi(w) \in J$, i.e., $w \in \STseq(L)$, i.e., $P \in L$.
    Finally, $M'$ is aperiodic since $[m]^{n+1} = [m^{n+1}] = [m^n] = [m]^n$. This concludes the proof.
\end{proof}

%% file: aperiodic.tex
\section{Aperiodic categories and counter-free HDA}
\label{sec:aperiodicity}

In the context of words, a deterministic automaton $\mathcal{A} = (\Sigma, Q, I, F, \delta)$ is \emph{counter-free} if
there exists $n \in \mathbb{N}$ such that, for any word $w \in \Sigma^*$ and state $q \in Q$,
$\delta(q, w^{n+1}) = \delta(q, w^n)$ \cite{McNaughton71}.
Essentially, this says that the transition monoid of the automaton is aperiodic, so the proof that counter-free automata recognize aperiodic languages is fairly straightforward.
It is, however, limited to deterministic automata, and it has been established that not all HDA can be determinized \cite[Section 6]{FahrenbergZ24}.
In \cite{DiekertG08}, Diekert and Gastin discussed two possible definitions of counter-freeness in the non-deterministic setting of Büchi automata. We adapt the less restrictive of them, Definition~11.5, for HDA.

\subsection{Counter-free HDA recognize aperiodic languages}

Given an HDA $X$, a cell $x \in X[U]$ and a pomset $P \in \iIPoms(U,V)$, we denote by $d(x, P)$ the set of cells accessible from $x$ by paths labeled by $P$.
\begin{definition}
\label{def:counter-free-hda}
    An HDA $X$ is \emph{counter-free} if there exists $n \in \mathbb{N}$ such that, for any cell $x \in X[U]$ and pomset $P \in \iIPoms(U, U)$,
    $d(x, P^{n+1}) = d(x, P^n)$.
\end{definition}

\begin{proposition}
\label{prop:counter-free-implies-aperiodic}
    For any counter-free HDA accepting a language $L$, there exists an aperiodic category recognizing $L$.
\end{proposition}
\begin{proof}
    If $L$ is recognized by a counter-free HDA, then the associated ST-automaton is counter-free. Hence, the word language $\STseq(L)$ is defined by an FO formula~\cite{McNaughton71}. Thus, so is the pomset language $L$~\cite{ClementEL26}.
    The conclusion comes from \Cref{thm:aper-cat-fo}.
\end{proof}

\begin{conjecture}
\label{con:aper-implies-counter-free}
	If $L\subseteq \iIPoms_{\leq k}$ is downward-closed w.r.t.\ subsumption, and is recognized be an aperiodic category, then $L$ is accepted by a counter-free HDA.
\end{conjecture}
This statement, which is the converse of \cref{prop:counter-free-implies-aperiodic}, is still open.
Though \cref{def:counter-free-hda} is very similar to other definitions of aperiodicity and counter-freeness, \cref{con:aper-implies-counter-free} is not immediate because of the geometric structure of HDAs imposing that any cell must be equipped with all possible faces.
%The semantic consequence is that any language accepted by an HDA is downward-closed with respect to subsumption,
The construction of \Cref{lem:pres-to-hda} that turns an $\iIPoms$-module into an HDA may need to add new cells to guarantee that face maps are properly defined.
While this construction does not change the accepted language (assuming it was already downward-closed), it can add non-trivial counters to the structure. %, effectively making the HDA non-counter-free.

\subsection{Creating counters}
\label{sec:creating-counters}

In this subsection, we show that the HDA construction of \Cref{lem:pres-to-hda} does not preserve aperiodicity.
%
%The idea relies on a similar construction to the one showing that not all pomset languages are recognized by a deterministic HDA from~\cite{FahrenbergZ24}.
%The reason why some languages cannot can be recognized only by non-deterministic HDA is because of their strict structure:
%if there is a path labelled by $\conclist{a\\b}$ from $x$ to $y$, then there must be a path labelled by $ab$ from $x$ to $y$.
%The semantic consequence of this is that any language accepted by an HDA is downward-closed. However, the structural implications are actually stronger.
%For instance, if $abc \in L$ but $\conclist{a\\b}\glue c \notin L$, then there must be a path labelled by $ab$ from the initial cell $x \in X$ to
%a cell $y$ frow which there is an accepting path labelled by $c$. But there must also be a path labelled by $\conclist{a\\b}$ from $x'$ to $y'$ such that
%$y'$ does not have an accepting path labelled by $c$. Hence, there is a path labelled by $ab$ from $x$ to $y'$. This is why the language
%$L = \{\conclist{a\\b}, ab, ba, abc\}$ cannot be accepted by any deterministic HDA, even though there is an $\iIPoms$-module, deterministic by nature,
%that recognizes $L$.
%
%The same issue arises when we try to show that any HDA constructed from an aperiodic module is counter-free.
Indeed, to turn an $\iIPoms$-module into an HDA recognizing the same language, the construction of \Cref{lem:pres-to-hda} (adapted from~\cite{FahrenbergZ24}) introduces supplementary lower faces.
Such extra cells, required by the HDA structure, may create counters that were not present in the original module.
\begin{proposition}
    There exists a counter-free presentation $(M,J,\varphi)$ such that the resulting HDA from \Cref{lem:pres-to-hda} is not counter-free.
\end{proposition}
\begin{proof}
    Consider the presentation $(M,J,\varphi)$ depicted in \Cref{fig:pres-aper:pres}.
    In this figure, elements of~$M$ are nodes in the graph, and the labels inside the nodes indicate the preimage by~$\phi$: when $\varphi(P) = m$, the pomset~$P$ is written within the node representing~$m$.
    Elements of $J$ are the nodes with a double edge.
    Arrows represent the action of some starter or terminator.
    This presentation recognizes the downward-closed language $L = \downclosure{\conclist{a\\b}} b^* + a$.
    It is aperiodic because the only cycles are $\varphi(ab) \cdot b = \varphi(ab)$ and $\varphi(abb\interface) \cdot \interface bb\interface
    = \varphi(abb\interface)$ and both are trivial.
    However, the HDA constructed from \Cref{lem:pres-to-hda},
    depicted in \Cref{fig:pres-aper:hda}, is not counter-free.
\end{proof}
\begin{figure}[h]
    \begin{subfigure}[b]{.60\textwidth}
        \centering
        \begin{tikzpicture}[scale=1.7, nodes={draw,align=center}, font=\footnotesize]
            \node at (0,0)                      (v00)   {$\varepsilon$};
            \node at (0,0.9)                      (v01)   {$b\interface$};
            \node at (0,1.8)                      (v02)   {$b$};
            \node at (0.9,0)                      (v10)   {$a\interface$};
            \node at (0.9,0.9)                      (v11)   {$\conclist{a\interface\\a\interface}$};
            \node at (0.9,1.8)                      (v12)   {$\conclist{a\interface\\b\pinterface},$\\$ba\interface$};
            \node [double] at (1.8,0)     (v20)   {$a, bab,$\\$\conclist{a\\b}\glue b$};
            \node at (1.8,0.9)                      (v21)   {$\conclist{a\pinterface\\b\interface}$};
            \node [double] at (1.8,1.8)     (v22)   {$\conclist{a\\b}, ba$};
            \node at (2.9,0)                      (v30)   {$ab\interface, babb\interface,$\\$\conclist{a\\b}\glue bb\interface$};
            \node at (2.9,0.9)                      (v31)   {$\conclist{a\\b}\glue b\interface,$\\$bab\interface$};
            \node [double] at (4,0)     (v40)   {$ab, \dots$};
            \node at (4,0.9)                      (v41)   {$abb\interface,\dots$};

            \draw[-latex, semithick]    (v00) -- (v01);
            \draw[-latex, semithick]    (v10) -- (v11);
            \draw[-latex, semithick]    (v01) -- (v02);
            \draw[-latex, semithick]    (v11) -- (v12);
            \draw[-latex, semithick]    (v21) -- (v22);
            \foreach \y in {0,1,2} {
                \draw[-latex, semithick]    (v0\y) -- (v1\y);
                \draw[-latex, semithick]    (v1\y) -- (v2\y);
            }
            \draw[-latex, semithick]    (v22) -- (v31);
            \draw[-latex, semithick]    (v31) -- (v20);
            \draw[-latex, semithick]    (v20) -- (v30);
            \draw[-latex, semithick]    (v30) -- (v40);
            \draw[-latex, semithick, bend left]    (v40) edge (v41);
            \draw[-latex, semithick, bend left]    (v41) edge (v40);
            \draw[-latex, semithick] (v00) -- (v11);
            \draw[-latex, semithick] (v11) -- (v22);
        \end{tikzpicture}
        \caption{A counter-free presentation for $L=\downclosure{\conclist{a\\b}}b^* + a$.}
        \label{fig:pres-aper:pres}
    \end{subfigure}
    \begin{subfigure}[b]{.35\textwidth}
        \centering
        \begin{tikzpicture}[scale=2, >=stealth, font=\footnotesize]
            % labels
            \node at (0.9,0.5)      {$b$};
            \node at (-0.1,0.5)     {$b$};
            \node at (0.5,-0.1)     {$a$};
            \node at (0.5,1.1)      {$a$};
            \node at (0.5,0.5)      {$\conclist{a\\b}$};
            \node at (1.23,0.5)     {$b$};
            \node at (1.5,0.1)      {$b$};
            \node at (2,0.55)       {$b$};
%            \node at (-0.12,-0.12)  {$\bot$};
%            \node at (1.12, -0.12)  {$\top$};
%            \node at (1.12,1.12)    {$\top$};
%            \node at (1.88,-0.12)   {$\top$};
            % square
            \fill[fill opacity = 0.1] (0,0) -- (0+1,0) --(0+1,0+1) --(0,0+1) -- cycle  ;
            % vertices
            \node[state, initial, initial text=, initial distance=6pt, fill=white, inner sep=0pt, minimum size=8pt] (v00) at (0,0) {};
            \node[state, fill=white, inner sep=0pt, minimum size=8pt] (v01) at (0,1) {};
            \node[state, accepting, fill=white, inner sep=0pt, minimum size=8pt] (v10) at (1,0) {};
            \node[state, accepting, fill=white, inner sep=0pt, minimum size=8pt] (v11) at (1,1) {};
            \node[state, accepting, fill=white, inner sep=0pt, minimum size=8pt] (v20) at (2,0) {};
            % edges
            \foreach \x in {0,1}
            {
                \draw[->, semithick] (v\x0) -- (v\x1);
            }
            \draw[->, semithick] (v00) -- (v01);
            \draw[->, very thick, red] (v10) -- (v11);
            \draw[->, semithick] (v00) -- (v10) ;
            \draw[->, semithick] (v10) -- (v20) ;
            \draw[->, semithick] (v01) -- (v11) ;
            \draw[->, very thick, bend left, red]  (v11) edge (v10) ;
            \draw[->, semithick, loop above] (v20) edge (v20) ;
        \end{tikzpicture}
        \caption{The resulting HDA, counter red.}
        \label{fig:pres-aper:hda}
    \end{subfigure}
    \caption{}%A presentation and the associated HDA for $L = \downclosure{\conclist{a\\b}} b^* + a$.}
    \label{fig:pres-aper}
\end{figure}

This counterexample shows that the most direct approach to proving \Cref{con:aper-implies-counter-free} does not work; however, it does not disprove the conjecture itself.
%This counterexample is not enough to disprove \Cref{con:aper-implies-counter-free}, which asserts that any language recognized by an aperiodic category must
%be accepted by some counter-free HDA. It simply asserts that the construction provided in \Cref{lem:pres-to-hda} does not preserve aperiodicity in some cases.
The language $L = \downclosure{\conclist{a\\b}} b^* + a$ can in fact be recognized by a counter-free HDA.
More precisely, the Myhill-Nerode HDA for this language, constructed from the canonical prefix presentation, is counter-free.
It is the HDA obtained from the one of \Cref{fig:pres-aper:hda} by merging the three accepting cells.

\subsection{Almost counting aperiodic language}
We now showcase an example of an aperiodic language that allows for a certain specific counting. We assume the one-letter alphabet $\Sigma=\{a\}$, and pomsets of dimension bounded by $k=2$.
Let us introduce pomsets $P_0=\varepsilon$, and for $n \geq 1$, 
\[
P_{2n}\;=\;a\interface\conclist{\interface a a\interface \\ a}^{n-1} \conclist{\interface a\\ \pinterface a}
\quad
=
\vcenter{\hbox{	\begin{tikzpicture}
		\node (0) at (0,1) {$a$};
		\node (1) at (0.5,0) {$a$};
		\node (2) at (1,1) {$a$};
		\node (3) at (1.5,0) {$a$};
		\node (4) at (2,1) {$a$};
		\node (5) at (2.5,0) {$a$};
		\draw [->] (0) -- (2);
		\draw [->] (1) -- (3);
		\draw [->] (2) -- (4);
		\draw [->] (3) -- (5);
		\draw [->] (0) -- (3);
		\draw [->] (1) -- (4);
		\draw [->] (2) -- (5);
		\draw [->,dotted] (0) -- (1);
		\draw [->,dotted] (2) -- (1);
		\draw [->,dotted] (2) -- (3);
		\draw [->,dotted] (4) -- (3);
		\draw [->,dotted] (4) -- (5);
		\node at (3.25,0) {$\dotsm$};
		\node at (2.75,1) {$\dotsm$};
		\node (a0) at (3.5,1) {$a$};
		\node (a1) at (4,0) {$a$};
		\node (a2) at (4.5,1) {$a$};
		\node (a3) at (5,0) {$a$};
		\node (a4) at (5.5,1) {$a$};
		\node (a5) at (6,0) {$a$};
		\draw [->] (a0) -- (a2);
		\draw [->] (a1) -- (a3);
		\draw [->] (a2) -- (a4);
		\draw [->] (a3) -- (a5);
		\draw [->] (a0) -- (a3);
		\draw [->] (a1) -- (a4);
		\draw [->] (a2) -- (a5);
		\draw [->,dotted] (a0) -- (a1);
		\draw [->,dotted] (a2) -- (a1);
		\draw [->,dotted] (a2) -- (a3);
		\draw [->,dotted] (a4) -- (a3);
		\draw [->,dotted] (a4) -- (a5);
	\end{tikzpicture}
}}
\]
Intuitively, $P_{2n}$ is the 2-dimensional pomset containing $2n$ events labeled by $a$ such that the $i$-th event is concurrent with
the $(i-1)$-th and the $(i+1)$-th events. Furthermore, they are alternating: odd events are on the top (i.e., they are minimal according to event order)
and even events are at the bottom (i.e., they are maximal).

\begin{lemma}
	The (non-downward-closed) language $\{P_{2n} \mid n\geq 0\}$ is aperiodic.
\end{lemma}
\begin{proof}
	It is easy to verify that this language is described as a set of pomsets $P$ satisfying the following first order formula.
	For every $x$, $a(x)$ and there exist $y, y'$ such that one of the following four conditions hold:
\begin{itemize}
\item $y \eventorder x \dashleftarrow y' \text{ and for all } z \neq x,y,y', \text{either } z < x \text{ or } x < z$,
\item $y \dashleftarrow x \eventorder y' \text{ and for all } z \neq x,y,y', \text{either } z < x \text{ or } x < z$,
\item $x \eventorder y \text{ and for all } z \neq x,y,\; x < z$,
\item $y \eventorder x \text{ and for all } z \neq x,y,\; z < x$. \qedhere
\end{itemize}
\end{proof}

The key idea is that we do not actually need to count the parity of the number of $a$'s. Instead, we keep track of whether the event is up or down according to the event order.
If we start up, alternate positions, and finish down, the number of $a$'s will be even.
Note that the language $\downclosure{\{P_{2n}\}}$ is not aperiodic, since $\downclosure{\{P_{2n}\}} \mathop{\cap} a^*=(aa)^*$ is not aperiodic.

\begin{lemma}
\label{l:weird-language}
	The language $L= \iIPoms_{\leq 2}(\emptyset, \emptyset) \setminus \{P_{2n}\}$ is downward-closed and aperiodic.
\end{lemma}
\begin{proof}
	The difference of aperiodic languages is aperiodic. Furthermore, $L$ is downward-closed  since $P_{2n}$ are maximal pomsets with respect to subsumption in $\iIPoms_{\leq 2}$.
\end{proof}

%\begin{corollary}
%	There exists an aperiodic downward-closed language $L\subseteq \iIPoms_{\leq 2}(\emptyset, \emptyset)$ such that
%	\begin{itemize}
%	\item 
%		it contains all pomsets of odd length,
%	\item 
%		for every $n\geq 0$, there is a pomset of length $2n$ that does not belong to $L$.
%\end{itemize}	
%\end{corollary}

Thus, there exists an aperiodic downward-closed language $L$ such that (1) it contains all pomsets of odd length, and (2) for every $n\geq 0$, there is a pomset of length $2n$ that does not belong to $L$.
We conclude from this example that aperiodic pomset languages admit a limited form of ``counting'', in contrast to word languages. The language $L$ from \cref{l:weird-language}, however, can be recognized by a counter-free HDA, so it still does not invalidate the conjecture.

%% file: conclusion.tex
\section{Conclusion and future work}

We have proved an algebraic characterization of regular languages of HDA: a pomset language is regular if and only if it is the preimage of a functor into a finite category.
Furthermore, the finite category recognizing a language is aperiodic if and only if the language is definable in first order logic.
This generalizes the McNaughton-Papert theorem to a higher-dimensional setting.
Aperiodic pomset languages are peculiar in that they allow for some restricted form of counting, unlike word languages.

We have defined a notion of counter-free HDA, and shown that every language accepted by a counter-free HDA is aperiodic.
However, the converse is still an open question: it is not yet well understood how aperiodicity interacts with downward-closure, which is intrinsic to HDAs.
One possible way to solve it would be to prove that the HDA constructed from the canonical suffix presentation using \Cref{lem:pres-to-hda} is always counter-free.
There does not seem to be a specific reason why this should be the case, but we have not yet found a counterexample.
Another, perhaps more promising, approach could be to show that, if an aperiodic language is recognized by a non-counter-free HDA, it is always possible to add some extra cells so as to eliminate the counter without changing the language.

One concrete use-case of the McNaughton-Papert theorem, for word automata, is to decide whether a given regular language is FO-definable: first compute its syntactic monoid, then check whether it is aperiodic. Doing the same for HDA is an interesting research direction: how to effectively compute the syntactic category, is it guaranteed to be aperiodic?
Finally, another avenue for future work would be to find a notion of star-free regular expressions that captures the same class of pomset languages.
HDA languages do not behave well w.r.t.\ complementation~\cite{AmraneBFZ25}, so finding a good notion of star-free regular expression is not straightforward.